\documentclass{eptcs}

\usepackage{amsmath}
\usepackage{amsfonts}
\usepackage{amsthm}
\usepackage{etex}
\usepackage{algorithm}
\usepackage{algorithmic}
\usepackage{booktabs}
\usepackage{bussproofs}
\usepackage{defs}
\usepackage{listings}
\usepackage{mathtools}
\usepackage{paralist}
\usepackage{stmaryrd}
\usepackage{subcaption}
\usepackage{array}
\usepackage{tikz}
\usepackage{xcolor}
\usepackage{xspace}

\usetikzlibrary{calc}
\usetikzlibrary{decorations.pathreplacing}
\usetikzlibrary{positioning}

\usepackage{letltxmacro}
\newcommand*{\SavedLstInline}{}
\LetLtxMacro\SavedLstInline\lstinline
\DeclareRobustCommand*{\lstinline}{%
  \ifmmode
    \let\SavedBGroup\bgroup
    \def\bgroup{%
      \let\bgroup\SavedBGroup
      \hbox\bgroup
    }%
  \fi
  \SavedLstInline
}

\lstdefinelanguage{ML}{
  alsoletter={*},
  morekeywords={datatype, of, if, *},
  sensitive=true,
  morecomment=[s]{/*}{*/},
  morestring=[b]"
}

\lstdefinelanguage{scala}{
  alsoletter={@,=,>},
  otherkeywords={passes},
  morekeywords={abstract, case, catch, choose, class, def, do, else, extends, false, final, finally, for, if, implicit, import, match, new, null, object, lazy, let,
override, package, private, protected, requires, return, sealed, super, this, throw, trait, try, true, type, val, var, while, with, yield, domain, 
postcondition, precondition, constraint, assert, forAll, _, return, @generator, ensure, require, ensuring, assuming, otherwise, asserting},
  sensitive=true,
  morecomment=[l]{//},
  morecomment=[s]{/*}{*/},
  morestring=[b]"
}

\newcommand{\codestyle}{\ttfamily}

\makeatletter
\newcommand*\idstyle{%
        \expandafter\id@style\the\lst@token\relax
}
\def\id@style#1#2\relax{%
        \ifcat#1\relax\else
                \ifnum`#1=\uccode`#1
                \fi
        \fi
}
\makeatother
\newcommand\ourtitle{On Repair with Probabilistic Attribute Grammars}         
\title{\ourtitle}

\author{
Manos Koukoutos
    \institute{EPFL}
    \email{emmanouil.koukoutos@epfl.ch}
\and
Mukund Raghothaman
\institute{University of Pennsylvania}
\email{rmukund@seas.upenn.edu}
\and
Etienne Kneuss
\institute{EPFL}
\email{etienne.kneuss@epfl.ch}
\and
Viktor Kuncak
\institute{EPFL}
\email{viktor.kuncak@epfl.ch}
}

\begin{document}

\maketitle

\begin{abstract}
Program synthesis and repair have emerged as an exciting area of research, driven by the potential for revolutionary advances in programmer productivity. Among most promising ideas emerging for synthesis are syntax-driven search, probabilistic models of code,
and the use of input-output examples. 
Our paper shows how to combine these techniques and use them for program repair,
which is among the most relevant applications of synthesis to general-purpose code.
Our approach combines semantic specifications, in the form of pre- and post-conditions and input-output examples with syntactic specifications in the form of term grammars and AST-level statistics extracted from code corpora. We show that synthesis in this framework can be viewed as an instance of graph search, permitting the use of well-understood families of techniques such as A*. We implement our algorithm in a framework for verification, synthesis and repair of functional programs, demonstrating that our approach can repair programs that are beyond the reach of previous tools.
\end{abstract}



\section{Introduction}
\label{sec:intro}

Program synthesis has attracted significant attention as a promising and ambitious approach
to improve software development productivity  \cite{Manna-1971,Sketch,Albarghouthi2013,Feser2015,CVC4-SyGuS,Brahma,Gulwani-etal-2011,InSynth,Osera-Myth-15, Osera-Myth-16,Synquid,LeonSynth-OOPSLA, LeonSynth-SYNT}. 
Among possibly deployment scenarios of synthesis is program repair, in which
a tool localized an error repairs the program by synthesizing an alternative
code fragment \cite{PeiETAL15RepairTests,LeonSynth-CAVRepair}.

Synthesis tools typically leverage both \emph{semantic} and \emph{syntactic} specification of the synthesis process. Typical forms of semantic information are specifications in the form of
input-output examples, as well as reference implementations and predicates. These specifications are important to identify the specific intent of the developer. Among the syntactic forms of specification are simple types and probabilistic models of code. Syntactic models of code are appealing for two main reasons: they can improve the efficiency of search, and they can provide implicit, reusable baseline specifications that apply across multiple synthesis problems, reducing the need to provide near-complete semantic specifications of the desired functionality.

Our work aims to unify these two sources of specifications. To this extent, we describe a specification framework and an algorithms that efficiently searches for programs based on:
\begin{itemize} 
\item a probabilistic and polymorphic attribute grammar that restricts and prioritizes the relevant space, and
\item dynamically growing set of input examples along with test predicates, which narrows
down the developer's intent.
\end{itemize}

Our attribute grammars describe the generation of trees and can be thought of as 
a probabilistic extension of type systems used in previous work. The use of generics
and attributes makes the grammars compact and enables us to naturally describe
constructs such as conditionals that are needed to generate more complex code fragments.
 Thanks to the fact
that we use a language supporting generics and value classes (Scala), we show how to describe
such derived grammars in the form of a wrapper library. The developers can therefore
use mostly existing programming language constructs to describe the intended models of
use of their libraries. The use of attributes enables the developers to incorporate
domain knowledge, such as associativity of operators, which can break symmetries and
reduce the search space. In addition to explicitly writing such models, we present an algorithm to automatically derive models from a corpus of code. Our results show that the derived
corpus provides benefits in synthesis of meaningful code fragments. 

Our algorithm efficiently supports specifications that include concrete inputs. In the
case where the provided set of inputs is not sufficient to narrow down the search, our system 
can also make use of constraint solvers to generate additional inputs. 
To make search driven by probabilistic grammars efficient, we adopt a top-down approach for search. We introduce techniques to prune the search for expressions using partial evaluation of incomplete grammars on known input examples. Furthermore, we show how to incorporate equivalence  modulo input-output examples in such a search algorithm, which is a technique that has so far been deployed primarily in bottom-up search techniques.

The practical focus of our techniques is on program repair. We use manually supplied
and automatically derived examples to localize the expression likely to be responsible, then synthesize an alternative 
expression in its place.
Examples
are particularly relevant because they can be used to localize the error
efficiently. The 
support of input examples in our algorithm makes it efficient for synthesizing
candidate repairs.
Furthermore, our support for probabilistic grammars enabled us to adopt the following approach
to repair with probabilistic models: 
in addition to using grammars describing general Scala expressions, we bias the search
for repairs towards the constructs specific for a program being repaired, making
the synthesis more effective. We therefore show that our approach is able to produce
program repairs whose size puts them beyond the reach of analogous existing 
non-probabilistic repair approaches.


\section{Synthesis Algorithm}
\label{sec:alg}



\subsection{Synthesis problem}

Let $\phi = \phi(\seqa, x)$ be the \emph{specification} of a programming task,
specifying a relation between a sequence of input parameters \seqa
and an output variable $x$.
$\phi$ can be given as a logical predicate, a set of input-output examples,
a reference implementation, or a combination thereof.
Also, let $\pcname = \pcname(\seqa)$ be the \emph{path condition} for the problem,
i.e. a relation satisfied by hypothesis by the input parameters.
Let also $t = t(\seqa)$ express a term in the target language.

The \emph{synthesis problem} is then defined as asking for a constructive solution for the formula

$$\forall \seqa. \exists t. \pcname \implies \phi[x \mapsto t]$$

In other words, we want to find an expression $t$ such that replacing every occurrence of the output variable $x$ in the
synthesis predicate $\phi$ causes the formula $\pcname \implies \phi$ to be \emph{valid}. In this section, we will
describe the probabilistic enumeration procedure used to solve such formulas.

\subsection{Counter-example guided inductive synthesis (CEGIS)}

The synthesis algorithm is an instantiation of CEGIS, a widely used framework to describe program synthesis algorithms~%
\cite{Sketch, Gulwani-etal-2011}. The main technical difficulty with finding a $t$ which satisfies rule $\TerminalRule$
is the universally quantified ``$\forall \seqa$'', where the input variables $\seqa$ may range over a large, or even
potentially infinite domain. We present the CEGIS routine in Algorithm~\ref{alg:cegis}. The procedure is parameterized
by two sub-procedures, $\Search$ and $\Verify$. At each step, we maintain a finite set $A$ of concrete input points. In
the search phase, the algorithm finds an expression $t$, such that
\begin{equation}
\Land_{\seqa \in A} \pcname \implies \phi[x \mapsto t].
\label{alg:cegis:search}
\end{equation}
Observe that for a specific choice of $A$, determining whether a given $t$ satisfies requirement~\ref{alg:cegis:search}
reduces to simply evaluating the resulting expression. In the verification phase, we confirm that the candidate
expression $t$ works for all inputs $\seqa$. The procedure $\Verify(\seqa, \pcname, \phi, x, t)$ returns either:
\begin{inparaenum}[(\itshape a\upshape)]
\item $\textsf{valid}$, if $t$ satisfies the synthesis predicate for all values of $\seqa$, or
\item a counter-example input $\seqa_{\text{cex}}$, such that $\pcname[\seqa \mapsto \seqa_{\text{cex}}] \land \lnot
    \phi[\seqa \mapsto \seqa_{\text{cex}}, x \mapsto t]$ i.e. $\seqa_{\text{cex}}$ satisfies the path condition $\pcname$, but $t$ fails the
  synthesis predicate.
\end{inparaenum}

\begin{algorithm}
\caption{$\CEGIS(\br{\seqa}{\pcname}{\phi}{x})$.}
\label{alg:cegis}
\begin{enumerate}
\item Initialize $A \coloneqq \emptyset$. $A$ is a finite set of concrete instantiations of $\seqa$.
\item Repeat forever:
  \begin{enumerate}
  \item Let $t = \Search(\seqa, A, \pcname, \phi, x)$.
  \item Let $\textsf{res} = \Verify(\seqa, \pcname, \phi, x, t)$.
  \item If $\textsf{res} = \textsf{valid}$, then return $t$.
  \item Otherwise, if $\textsf{res} = \seqa_{\text{cex}}$, update $A \coloneqq A \union \{ \seqa_{\text{cex}} \}$.
  \end{enumerate}
\end{enumerate}
\end{algorithm}

We discharge the calls to $\Verify$ by invoking Leon's verification system~\cite{Leon-Scala13}. The main technical
contribution of this paper is in using probabilistic enumeration to instantiate the $\Search$ subroutine. We devote the
rest of this section to describing this procedure.

\subsection{Probabilistic expression grammars}

The goal of the $\Search(\seqa, A, \pcname, \phi, x)$ subroutine is to find a candidate expression $t$ which satisfies
the synthesis predicate for all input points $\seqa \in A$. The expression $t$ is drawn from the productions of a
probabilistic context-free grammar (PCFG), such as that shown below:
\[
\begin{array}{rcll}
  \textsf{IntExpr} & \Coloneqq & \lstinline|0| & (\text{Rule } R_0, p_0 = 0.15) \\
                   & \mid & \lstinline|1| & (\text{Rule } R_1, p_1 = 0.3) \\
                   & \mid & \lstinline|x| & (\text{Rule } R_x, p_x = 0.3) \\
                   & \mid & \textsf{IntExpr} + \textsf{IntExpr} & (\text{Rule } R_+, p_+ = 0.15) \\
                   & \mid & \lstinline|if ($\textsf{BoolExpr}$) { $\textsf{IntExpr}$ } else { $\textsf{IntExpr}$ }|
                          & (\text{Rule } R_{\text{ite}}, p_{\text{ite}} = 0.1) \\
  \textsf{BoolExpr} & \Coloneqq & \lstinline|$\textsf{IntExpr}$ $\leq$ $\textsf{IntExpr}$|
                                & (\text{Rule } R_\leq, p_\leq = 0.8) \\
                    & \mid & \lstinline|$\textsf{BoolExpr}$ && $\textsf{BoolExpr}$|
                           &  (\text{Rule } R_{\land}, p_\land = 0.2)
\end{array}
\]
The productions of $\textsf{IntExpr}$ expand to expressions of type \lstinline|Int|, and the productions of
$\textsf{BoolExpr}$ expand to expressions of type \lstinline|Bool|. Each rule $R$ of each non-terminal symbol $N$ is
annotated with a number $p_R$, which indicates the probability of $N$ expanding to $R$. Naturally, the probability $p_e$
of a production $e$ is then defined as the product of the probabilities of all rules invoked in the derivation tree: for
example, the production \lstinline|x + 1| has probability $p_+ p_x p_1 = 0.0135$.

More generally, a PCFG is a pair $G = (\mathcal{N}, \mathcal{R})$, where:
\begin{enumerate}
\item $\mathcal{N}$ is a finite, non-empty set of non-terminal symbols, where each non-terminal $N \in \mathcal{N}$ is
  associated with a type $T_N$,
\item $\mathcal{R}$ maps each non-terminal $N$ to a finite set of production rules $\mathcal{R}(N)$. Each production
  rule $R \in \mathcal{R}(N)$ is a well-typed construct of the form $R = t(N_1, N_2, \ldots, N_k)$, where $t$ is the
  top-level operator, $N_1$, $N_2$, \ldots, $N_k$ are the child non-terminal symbols, and such that the output types of
  $t$ and $N$ coincide: $T_t = T_N$, and
\item each rule $R$ is associated with a probabiltiy $p_R \in [0, 1]$ such that for all non-terminals $N$, $\sum_{R \in
  \mathcal{R}(N)} p_R = 1$.
\end{enumerate}

Notice that $\mathcal{N}$ does not necessarily coincide with the set of types of the target language;
we only require that the mapping from non-terminals to types is surjective.

We make two additional requirements of the PCFGs we consider in this paper: first, we require the probability of each
rule, $p_R < 1$ (note the strict inequality), and second, we require that $G$ be \emph{unambiguous}, i.e. that every
expression produced by the grammar have a unique derivation tree. If we relax the first assumption, then there are some
technical difficulties in setting up the probability space, and results such as Theorem~\ref{thm:alg:Dijkstra} will need
additional constraints to be true. The second assumption simplifies the definition of $P(e)$, which would otherwise be
the sum of the probabilities over all derivation trees. Moreover, it reduces the search space to a tree (instead of a
DAG), and thereby allows us to use stronger versions of the guarantees provided by search algorithms.

For simplicity of notation, and whenever the intent is clear, we will conflate productions $e$ of the grammar with the
expressions $t_e$ that they encode. We write $\Expansions_N$ for the set of all productions of a non-terminal $N$. Given
a production $e \in \Expansions_N$, we define its probability, $P(e)$, as the product of the probabilities of all rules
used in $e$. It can be shown that $\sum_{e \in \Expansions_N} P(e) = 1$. The productions of $G$ will constitute elements
of the sample space under consideration, and the space of events is, as usual, the powerset $2^{\Expansions_N}$ of
productions. We will discuss the extraction of PCFGs from a code corpus in Section~\ref{sec:grammar-gen}.

\subsection{Top-down expression enumeration as graph search}

In the search phase of the CEGIS loop, we want to find an expression $t$ which satisfies the requirement in formula~%
\ref{alg:cegis:search}. One approach to implementing $\Search(\seqa, A, \pcname, \phi, x)$ is to enumerate all candidate
expressions, in some order, until a suitable answer is found. Recall that, because $A$ is a finite set, for a given
choice of $A$ and $t$, determining the satisfaction of formula~\ref{alg:cegis:search} reduces to evaluating the
conjunction, and does not involve expensive calls to an SMT solver. The enumerative SyGuS solver~\cite{Transit} uses a
highly optimized form of expression enumeration in the search phase, and is currently among the most competitive SyGuS
solvers. See Algorithm~\ref{alg:search}. In this section, we will show how we utilize a PCFG
to enumerate expressions in order of decreasing probability in order to accelerate the search
process. We will now describe this probabilistic enumeration algorithm.

\begin{algorithm}
\caption{$\Search(\seqa, A, \pcname, \phi, x)$. Implements the search phase of the CEGIS loop by expression
  enumeration.}
\label{alg:search}
\begin{enumerate}
\item Let $G$ be the chosen PCFG, and $N$ be the starting non-terminal such that the types coincide, i.e. $T_N = T_x$.
\item For each $e$ emitted by $\Enumerate(G, N)$:
  \begin{enumerate}
  \item Let $t_e$ be the expression encoded by $e$.
  \item If $\Land_{\seqa \in A} \pcname \implies \phi[x \mapsto t_e] = \lstinline|true|$, return $t_e$.
  \item Otherwise, discard $t$ and continue enumeration.
  \end{enumerate}
\end{enumerate}
\end{algorithm}

\paragraph{Partial productions.}
Engineering the enumeration algorithm turns out to be much simpler if we enumerate expressions top-down rather than
bottom-up. As a first step, we therefore extend the idea of a grammar production into the more general notion of a
\emph{partial} production:
\[
\begin{array}{rcl}
  \tilde{e}_N & \Coloneqq & \lstinline|?|_N \mid
                      t(\tilde{e}_{N_1}, \tilde{e}_{N_2}, \ldots, \tilde{e}_{N_k}),
\end{array}
\]
where $R = t(N_1, N_2, \ldots, N_k)$ is a production rule of $N$ in $G$. In particular, notice that partial productions
may contain incomplete sub-expressions, denoted by $\lstinline|?|_N$. Examples of partial productions include
``\lstinline|x + ?|'' and ``\lstinline|if (x $\leq$ ?) { x } else { ? }|''. We write $\tilde{\Expansions}_N$ for the set
of all partial productions of a non-terminal $N$. We can also speak of the probability of partial productions:
$p_{\tilde{e}}$ of a partial production $\tilde{e}$ is the product of all production rules used to derive $\tilde{e}$.
It will be mathematically more convenient to speak of negative log probabilities: the \emph{cost} of a partial
production $\tilde{e}$,
\begin{equation}
\cost(\tilde{e}) = -\log(p_{\tilde{e}}) = -\sum_{R \in \tilde{e}} \log(p_R).
\label{eq:alg:costdefn}
\end{equation}

Given a pair of partial productions $\tilde{e}_1$ and $\tilde{e}_2$, we say that $\tilde{e}_1$ and $\tilde{e}_2$ are
related by the expansion relation, and write $\tilde{e}_1 \to \tilde{e}_2$, if $\tilde{e}_2$ can be obtained by
replacing the left-most instance of \lstinline|?| in $\tilde{e}_1$ by an appropriate production rule. For example, if
$\tilde{e}_1 = \lstinline|x + ?|$, $\tilde{e}_2 = \lstinline|x + 1|$, and $\tilde{e}_3 = \lstinline|x + (? + ?)|$, then
$\tilde{e}_1 \to \tilde{e}_2$ and $\tilde{e}_1 \to \tilde{e}_3$.

The expansion relation naturally induces a tree $\mathcal{G}$ on the partial productions of a grammar $G$. See Figure~%
\ref{fig:alg:graph}. The initial node is the starting non-terminal, \lstinline|?|. There is an edge from the partial
production $\tilde{e}_1$ to the partial production $\tilde{e}_2$ if they are related by the expansion relation,
$\tilde{e}_1 \to \tilde{e}_2$. The edge $\tilde{e}_1 \to \tilde{e}_2$, produced by an instantiation of the rule $R$, is
annotated with the cost of the rule, $-\log(p_R)$. Then, the resulting graph is a tree because of the unambiguity
assumption on $G$, and the sum of the weights along the path to $\tilde{e}$ is equal to $\cost(\tilde{e})$. The main
insight of this paper is that $\Enumerate(G, N)$ can be thought of as instantiations of various search algorithms on
this tree.

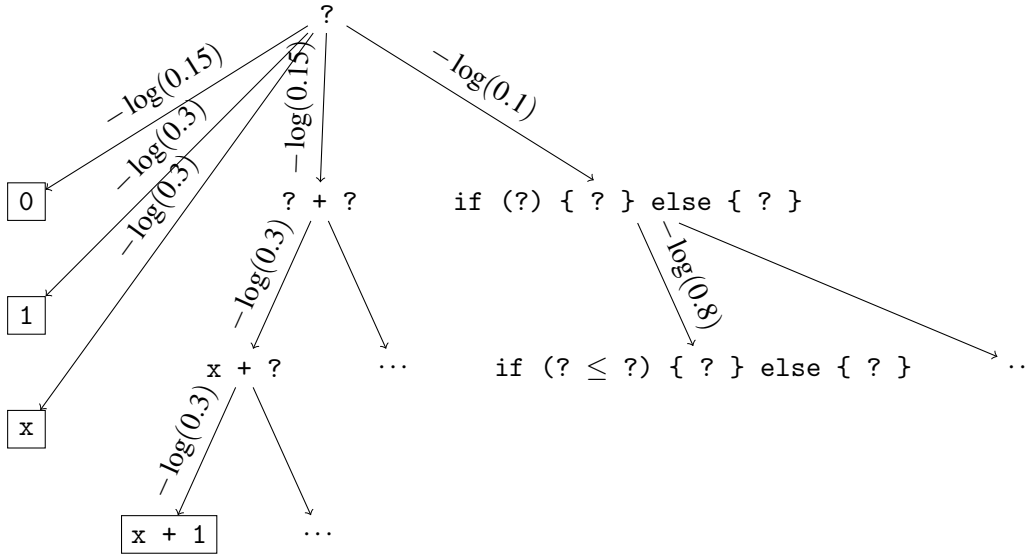
\begin{figure}
\begin{centering}
\begin{tikzpicture}[minimum size=0.5 cm]

\node [draw] (e1) at (-2, -1) {\lstinline|0|};
\node [draw, below=of e1] (e2) {\lstinline|1|};
\node [draw, below=of e2] (e3) {\lstinline|x|};
\node [right=3 of e1] (e4) {\lstinline|? + ?|};
\node [right=of e4] (e5) {\lstinline|if (?) { ? } else { ? }|};

\node (e0) at ($ (e1)!0.5!(e5) + (0, 2.5) $) {\lstinline|?|};

\path [->] (e0) edge [above, sloped] node {$-\log(0.15)$} (e1);
\path [->] (e0) edge [above, sloped] node {$-\log(0.3)$} (e2);
\path [->] (e0) edge [above, sloped] node {$-\log(0.3)$} (e3);
\path [->] (e0) edge [above, sloped] node {$-\log(0.15)$} (e4);
\path [->] (e0) edge [above, sloped] node {$-\log(0.1)$} (e5);

\node (e6) at ({$(e4) + (-1, 0)$} |- {$(e2) + (0, -0.7)$}) {\lstinline|x + ?|};
\path [->] (e4) edge [above, sloped] node {$-\log(0.3)$} (e6);
\node (e7) at ({$(e4) + (1, 0)$} |- {$(e2) + (0, -0.7)$}) {$\cdots$};
\path [->] (e4) edge (e7);

\node [draw] (e8) at ({$(e6) + (-1, 0)$} |- {$(e3) + (0, -1.4)$}) {\lstinline|x + 1|};
\path [->] (e6) edge [above, sloped] node {$-\log(0.3)$} (e8);
\node (e9) at ({$(e6) + (1, 0)$} |- {$(e3) + (0, -1.4)$}) {$\cdots$};
\path [->] (e6) edge (e9);

\node (e10) at ({$(e5) + (1, 0)$} |- e7) {\lstinline|if (? $\leq$ ?) { ? } else { ? }|};
\path [->] (e5) edge [above, sloped] node {$-\log(0.8)$} (e10);
\node [right=of e10] (e11) {$\cdots$};
\path [->] (e5) edge (e11);

\end{tikzpicture}
\end{centering}
\caption{The tree $\mathcal{G}$ of partial productions, connected by the expansion relation. The initial node is the
  empty partial production \lstinline|?| of the non-terminal \lstinline|IntExpr|. Each edge indicates the replacement of
  the left-most incomplete sub-production by an instance of a production rule $R$. The edges are annotated with the
  negative log probabilities, $-\log(p_R)$: because $0 \leq p_R \leq 1$, $\log(p_R) \leq 0$, so the weights $-\log(p_R)$
  are non-negative. A production is enclosed in a rectangle if it is complete.}
\label{fig:alg:graph}
\end{figure}

\begin{algorithm}
\caption{$\Enumerate(G, N)$. Emits a sequence of complete productions of $N$.}
\label{alg:enum}
\begin{enumerate}
\item Let $\pi : \tilde{\Expansions}_N \to \R^{\geq 0}$ be a priority function mapping productions to non-negative real
  numbers.
\item Let $Q$ be a priority queue of partial productions $\tilde{e} \in \tilde{\Expansions}_N$ arranged in ascending
  order according to $\pi$. Initialize $Q \coloneqq \{ \lstinline|?|_N \}$.
\item While $Q$ is not empty:
  \begin{enumerate}
  \item \label{enum:alg:enum:dequeue} Let $\tilde{e}$ be the element at the front of $Q$. Dequeue $\tilde{e}$.
  \item If $\tilde{e}$ is a complete production, emit $\tilde{e}$.
  \item Otherwise, for every neighbor $\tilde{e}'$ such that $\tilde{e} \to \tilde{e}'$ in $\mathcal{G}$, insert
    $\tilde{e}'$ into $Q$.
  \end{enumerate}
\end{enumerate}
\end{algorithm}

See Algorithm~\ref{alg:enum}. It is helpful to visualize the operation of the algorithm as shown in Figure~%
\ref{fig:alg:enum}. The algorithm maintains a priority queue $Q$ of still-unexpanded partial productions. At each step,
the algorithm dequeues the element $\tilde{e}$ at the front of this priority queue and, if it is still incomplete,
expands the leftmost \lstinline|?| with an instance of every applicable production rule $R$. This results in a set of
partial productions $\tilde{e}_1$, $\tilde{e}_2$, \ldots, $\tilde{e}_k$. All of these new partial productions are
inserted back into $Q$. If, on the other hand, $\tilde{e}$ is already complete, it is emitted to be further processed by
$\Search(\seqa, A, \pcname, \phi, x)$. If $\pi(\tilde{e}) = \cost(\tilde{e}) = -\log(P(\tilde{e}))$, i.e. the priority
function is equal to the cost of the partial production, then $Q$ forms the frontier maintained by Dijkstra's algorithm.

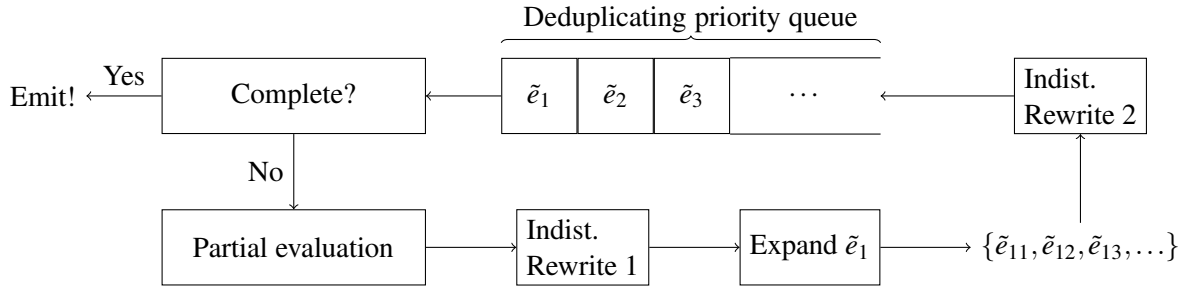
\begin{figure}
\resizebox{\textwidth}{!}{
\begin{tikzpicture}
  \node (emit) {Emit!};
  \node [minimum height=1cm, minimum width=3.5cm, draw, right=of emit] (completeCheck) {Complete?};
  \path [->] (completeCheck) edge [above] node {Yes} (emit);

  \node [draw, minimum size=1cm, right=of completeCheck] (e1) {$\tilde{e}_1$};
  \node [draw, minimum size=1cm, right=0cm of e1] (e2) {$\tilde{e}_2$};
  \node [draw, minimum size=1cm, right=0cm of e2] (e3) {$\tilde{e}_3$};
  \draw (e3.north east) -- +(2, 0);
  \draw (e3.south east) -- +(2, 0);
  \path [->] (e1) edge (completeCheck);

  \node (edots) at ($ (e3.east) + (1, 0) $) {$\cdots$};

  \node [draw, minimum height=1cm, minimum width=3.5cm, below=of completeCheck] (peval) {Partial evaluation};
  \path [->] (completeCheck) edge [left] node {No} (peval);
  \node [draw, minimum height=1cm, align=left, right=1.2 of peval] (indist1) {Indist. \\ Rewrite 1};
  \path [->] (peval) edge (indist1);
  \node [draw, minimum height=1cm, right=1.2 of indist1] (expand) {Expand $\tilde{e}_1$};
  \path [->] (indist1) edge (expand);
  \node [right=1.2 of expand] (echildren) {$\{ \tilde{e}_{11}, \tilde{e}_{12}, \tilde{e}_{13}, \ldots \}$};
  \path [->] (expand) edge (echildren);
  \node [draw, minimum height=1cm, align=left] (indist2) at ( echildren |- edots) {Indist. \\ Rewrite 2};
  \path [->] (echildren) edge (indist2);
  \path [->] (indist2) edge ($ (e3.east) + (2, 0) $);

  \draw [decorate, decoration={brace}] ($ (e1.north west) + (0, 0.2) $)
                                    -- ($ (e3.north east) + (2, 0.2) $)
        node [midway, above] {Deduplicating priority queue};
\end{tikzpicture}
}
\caption{The operation of $\Enumerate(G, N)$ in Algorithm~\ref{alg:enum}. We describe the partial evaluation and
  indistinguishability-based optimizations, including production rewriting and queue deduplication, in
  Section~\ref{sub:alg:optimizations}.}
\label{fig:alg:enum}
\end{figure}

We begin analysing Algorithm~\ref{alg:enum} with the following invariant: Let $\Expansions_c$ be the set of (complete)
productions emitted by $\Enumerate(G, N)$. Then, $\sum_{\tilde{e} \in Q \union \Expansions_c} P(\tilde{e}) = 1$. We will
now outline the proof of its invariance. First, observe that, for all partial productions $\tilde{e}$ with at least one
occurrence of \lstinline|?|,
\begin{equation}
  \sum_{\tilde{e}' \text{s.t.} \tilde{e} \to \tilde{e}'} P(\tilde{e}') = P(\tilde{e}).
  \label{eq:alg:inv:sum1}
\end{equation}
Informally, if each partial production $\tilde{e}$ is viewed as an event resulting from the PCFG, then the partial
productions $\tilde{e}'$ obtained by a single application of the expansion relation encode a mutually exclusive and
exhaustive collection of sub-events of $\tilde{e}$. We then have:
\begin{inparaenum}[(\itshape a\upshape)]
\item the base case: $P(\lstinline|?|) = 1$,
\item inductive case \#1: whenever a complete production $e$ is emitted from $Q$, the set $Q \union \Expansions_c$ is
  left unchanged,
\item and inductive case \#2: whenever a partial production $\tilde{e}$ is further expanded into partial productions
  $\tilde{e}_1$, $\tilde{e}_2$, \ldots, $\tilde{e}_k$, the candidate invariant continues to hold because of
  Equation~\ref{eq:alg:inv:sum1}.
\end{inparaenum}

\begin{thm}
\label{thm:alg:dijkstra}
If the priority function, $\pi(\tilde{e}) = \cost(\tilde{e}) = -\log(P(\tilde{e}))$, then $\Enumerate(G, N)$ returns
some sequence $\sigma$ containing all productions of $N$ in $G$. Furthermore, the probability $P(e)$ of productions in
$\sigma$ is monotonically decreasing.
\end{thm}
\begin{proof}
We begin by proving the second part of the claim. Consider a pair of partial productions, $\tilde{e}$ and $\tilde{e}'$
such that $\tilde{e} \to \tilde{e}'$. Then, $P(\tilde{e}) \geq P(\tilde{e}')$. It therefore follows that the sequence of
partial productions, $\tau$, dequeued by $\Enumerate(G, N)$ in line~\ref{enum:alg:enum:dequeue} has monotonically
decreasing probability. Observe that the sequence $\sigma$ of complete productions emitted by $\Enumerate(G, N)$ is a
sub-sequence of $\tau$, and therefore, the probability of productions in $\sigma$ monotonically decreases.

Now, arbitrarily choose a complete production $e \in \Expansions_N$. We will show that $\Enumerate(G, N)$ will
eventually emit $e$. Let $p_{\text{hi}}$ be the highest probability of any production rule in $G$. Recall our
requirement that for each rule $R$, $p_R < 1$, and therefore, $p_{\text{hi}} < 1$ and $\log(p_{\text{hi}}) < 0$.
Before processing $e$, $\Enumerate(G, N)$ can only process those nodes which are at most $\log(P(e)) /
\log(p_{\text{hi}})$ steps away from the root node \lstinline|?| of the search tree $\mathcal{G}$. If there are $N$
rules in the PCFG, then the out-degree of each node is at most $N$, and it follows that there are only finitely many
nodes which can be processed before $\Enumerate(G, N)$ processes $e$. It follows that $e$ is eventually emitted by the
enumerator.
\end{proof}

Unfortunately, when the priority function is instantiated to emulate Dijkstra's algorithm, enumeration does not scale
beyond a few hundred productions. The intuition may be found in the correctness proof above: before emitting $e$, the
enumerator must explore all partial productions with probability bigger than $P(e)$. The number of these productions
rapidly grows with decreasing $P(e)$, and furthermore, most of the processed partial productions are ``very
incomplete'', i.e. containing many instances of \lstinline|?|, and therefore many edges away from turning into complete
productions. This suggests other graph search algorithms, which we will now discuss.


\paragraph{Non-terminal horizons and A* search.}
When considering partial productions $\tilde{e}$, we may speak of two quantities:
\begin{inparaenum}[(\itshape a\upshape)]
\item the cost \emph{already paid}, which we have defined as $\cost(\tilde{e}) = -\log(P(\tilde{e}))$, and
\item the minimum cost \emph{yet to be paid}, before $\tilde{e}$ turns into a complete production.
\end{inparaenum}
We formalize this latter quantity as the \emph{horizon}, defined as:
\begin{equation}
  \horizon(\tilde{e}) = \min_{e_f \text{ s.t. } \tilde{e} \to^* e_f} \cost(e_f) - \cost(\tilde{e}),
  \label{eq:alg:horizon}
\end{equation}
where $e_f$ ranges over all complete productions reachable from $\tilde{e}$. Recall the intuition for the failure of
Dijkstra's algorithm: before emitting $e$, we have to process every partial production with higher probability, and most
such partial productions $\tilde{e}'$ are themselves many steps away from complete productions. Since the horizon
encodes the distance from the partial production to its nearest completion, it is natural to include it in the priority
function $\pi(\tilde{e})$. If we define $\pi(\tilde{e}) = \cost(\tilde{e}) + \horizon(\tilde{e})$, then we obtain A*
search.

The important property of $\pi(\tilde{e})$ is that it is \emph{admissible}, i.e. that $\pi(\tilde{e})$ is always less
than the cost of all complete productions descendant from $\tilde{e})$. More formally, for all complete productions
$e_c$ reachable from $\tilde{e}$,
\begin{equation}
  \pi(\tilde{e}) = \min_{e_c \text{s.t.} \tilde{e} \to^* e_c} \cost(e_c).
  \label{eq:alg:astar}
\end{equation}
We then have the following well-known result of the A* algorithm:

\begin{thm}
\label{thm:alg:astar}
If $\pi(\tilde{e}) = \cost(\tilde{e}) + \horizon(\tilde{e})$, then:
\begin{enumerate}
\item The sequence of complete productions emitted by $\Enumerate(G, N)$, $\sigma = e_1, e_2, e_3, \ldots$ contains all
  productions in $\Expansions_N$ and has monotonically decreasing probability. (This is called the \emph{optimality}
  property.)
\item For all complete productions $e$, every complete optimal algorithm $A'(G, N)$ will, before producing $e$, visit
  all strictly less expensive nodes, $\tilde{e}$, such that $\pi(\tilde{e}) < \pi(e)$ . (This is the property of
  \emph{optimal efficiency}.)
\end{enumerate}
\end{thm}
\begin{proof}
First, we show that $\Enumerate(G, N)$ produces all complete productions $e \in \Expansions_N$. Similar to Theorem~%
\ref{thm:alg:dijkstra}, we bound the number of nodes that can be processed in line~\ref{enum:alg:enum:dequeue} before
processing $e$. Observe that for all nodes, $\tilde{e}$, $\pi(\tilde{e}) \geq \cost(\tilde{e})$, and for complete
productions $e$,  $\pi(e) = \cost(e)$. Therefore, all nodes processed by $\Enumerate(G, N)$ before processing $e$ are
no more than $\log(P(e)) / \log(p_{\text{hi}})$ steps from the root of the search tree $\mathcal{G}$, where, as before,
$p_{\text{hi}} < 1$ is the highest probability appearing in the PCFG.

We will now show that the productions emitted have monotonically decreasing probability. For this, first observe that
for each complete production $e$, and for each of its (necessarily incomplete) ancestors $\tilde{e}'$, $\pi(\tilde{e}')
\leq \pi(e)$. For the sake of contradiction, let there be some pair $e_1$, $e_2$ of complete productions, such that
$P(e_1) < P(e_2)$, but $\Enumerate(G, N)$ emits $e_1$ before $e_2$. Both $e_1$ and $e_2$ are therefore reachable from
the initial node \lstinline|?|. Consider the state of $Q$ when $e_1$ is dequeued from it. It has to be the case that
some ancestor $\tilde{e}'_2$ is present in $Q$. However, it follows that $\pi(\tilde{e}'_2) \leq \pi(e_2) < \pi(e_1)$,
and therefore the priority queue must have made a mistake in its ordering. It follows that the sequence of productions
emitted by $\Enumerate(G, N)$ has monotonically increasing costs, or equivalently, monotonically decreasing
probabilities.

We will now prove the last part of the theorem. Assume otherwise, so for some algorithm $A'(G, N)$ and some node
$\tilde{e}$ such that such that $\pi(\tilde{e}) < \pi(e)$, $A'$ does not expand $\tilde{e}$ before producing $e$.
We know $\tilde{e}$ has some descendant $e'$ such that $\pi(\tilde{e}) = \cost(e') < \pi(e)$. Because the search space
$\mathcal{G}$ is a tree, if $A'$ does not expand $\tilde{e}$, it follows that it did not enumerate $e'$ before $e$,
violating the assumption that $A'$ is a complete optimal enumerator.
\end{proof}

Note that, in a certain sense, the second part of Theorem~\ref{thm:alg:astar} says that this is the best possible
enumeration algorithm among all algorithms that use the same priority function, because $\Enumerate(G, N)$ does not
expand any nodes $\tilde{e}'$ such that $\pi(\tilde{e}') > \pi(e)$. This also captures our intuition for the superiority
of this priority function over the simpler metric used by Dijkstra's algorithm.

\paragraph{Computing the horizon, $\horizon(\tilde{e})$.}
It can be shown that the cost to be paid to complete $\tilde{e}$ is the sum of the costs needed to expand each of its
unexpanded nodes:
\begin{equation}
  \horizon(\tilde{e}) = \sum_{\lstinline|?|_N \in \tilde{e}} \horizon(\lstinline|?|_N).
\end{equation}
Before starting the enumerator, we therefore compute the non-terminal horizons, $\horizon(\lstinline|?|_N)$, for each
non-terminal symbol $N$. Observe that, by definition, this simply encodes the probability of the most likely expansion,
$e \in \Expansions_N$:
\[
  \horizon(\lstinline|?|_N) = \min_{e \in \Expansions_N} \cost(e) = -\max_{e \in \Expansions_N} \log(P(e)).
\]
The values $\horizon(\lstinline|?|_N)$ can be computed by the following fixpoint computation:
\begin{enumerate}
\item $\horizon^0(\lstinline|?|_N) = \min \{ -\log(p_R) \mid \text{Terminal rule } R \}$.
\item $\horizon^{i + 1}(\lstinline|?|_N) = \min \{ \horizon^i(R) \mid \text{Production rule } R \}$, where the horizon
  of a production rule $R$ is defined as $\horizon^i(R) = -\log(p_R) + \sum \{ \horizon^i(\lstinline|?|_{N'}) \mid
  \text{child non-terminal } N' \text{ appearing in } R \}$.
\end{enumerate}

\subsection{Optimizations}
\label{sub:alg:optimizations}

\paragraph{Eagerly discarding partial productions.}
Given a single input variable \lstinline|a|, consider the synthesis predicate $\phi$ given by:
\begin{lstlisting}
  if (a == 5) { x == 6 }
  else if (a == 7) { x == 9 }
  else { x == a }
\end{lstlisting}
During CEGIS, say the set of concrete input points, $A = \{ 2, 5, 7 \}$. For this problem, $\Enumerate(G, N)$ will
consider partial productions such as
\[
  \tilde{e} = \lstinline|if (x < 6) { x } else { ? }|.
\]
Observe that the conjunction $\Land_{\lstinline|a| \in A} \phi[\lstinline|x| \mapsto e]$ used in the CEGIS loop can only
be evaluated for complete productions $e$. However the partial production $\tilde{e}$ is already incorrect: for the
input point $\lstinline|a| = \lstinline|5|$, regardless of the completion of \lstinline|?|, $\tilde{e}$ will evaluate to
\lstinline|5|, and this fails the requirement that \lstinline|x == 6|.

This pattern of partial productions already failing requirements is particularly common with conditionals and match
statements. Our first optimization is therefore the partial evaluation box shown in Figure~\ref{fig:alg:enum}. For each
input point $\seqa \in A$, we partially evaluate the CEGIS predicate $\pcname \implies \phi[x \mapsto \tilde{e}]$. If
the result is \lstinline|false|, then we discard the partial production $\tilde{e}$ without processing, as we know that
all its descendants will fail the synthesis predicate. Otherwise, if it evaluates to \lstinline|true| or unknown, then
we process $\tilde{e}$ as usual, and insert its neighbors back into the priority queue.

\paragraph{Extending the priority function with scores.}
We just observed that if a partial production fails the synthesis predicate by partial evaluation, then it can be
discarded without affecting correctness. The dual heuristic optimization is to \emph{promote} partial productions which
definitely evaluate to \lstinline|true| on some input points. We do this by modifying the priority function:
\begin{equation}
  \pi(\tilde{e}) = \cost(\tilde{e}) + \horizon(\tilde{e}) + c \times \log(\score(\tilde{e})),
\end{equation}
where $\score(\tilde{e})$ is the number of input points on which the partial production evaluates to \lstinline|true|,
and $c$ is a positive coefficient. Small positive values of $c$ results in an enumerator which strictly follows
probabilistic enumeration, while large positive values of $c$ results in the enumerator favoring partial productions
which already work on many points. While the use of this heuristic renders Theorem~\ref{thm:alg:astar} inapplicable, we
have observed improvements in performance as a result of this optimization.

\paragraph{Indistinguishability.}
Consider the partial production $\tilde{e} = \lstinline|if (x < 5) { x - x } else { ? }|$, and focus on the
sub-expression \lstinline|x - x|. This expression is identically equal to \lstinline|0|, and it is therefore possible
to simplify $\tilde{e}$ to \lstinline|if (x < 5) { 0 } else { ? }|. If two expressions $e$ and $e'$ are equivalent, and
$P(e) > P(e')$, then a desirable optimization is to never enumerate $e'$. Observe that, if properly implemented, this
optimization produces exponential savings at each step of the search: if $e$ and $e'$ are equivalent, then, for example,
it follows that both $e + 3$ and $e' + 3$ are equivalent, and only one of them needs to be enumerated to achieve
completeness.

Indistinguishability~\cite{Transit} is a technique to mechanize this reasoning. Given a sequence of concrete input
points $A = \{ \seqa_1, \seqa_2, \ldots, \seqa_n \}$, a complete production $e$ can be evaluated to produce a set of
output values $S_e = \{ x_1, x_2, \ldots, x_n \}$. During enumeration, if an expression $e'$ is encountered such that
for some previous expression $e$, $S_e = S_{e'}$, then $e'$ is discarded as a potential expansion. The original
expression $e$ can thus be regarded as the representative of the equivalence class of all expressions whose signature is
equal to $S_e$.

In this original formulation of indistinguishability-based pruning, the enumerator worked bottom-up, rather than in the
top-down style we consider in this paper. As a result, the enumerator only needs to process complete productions, rather
than the partial productions which populate $Q$ in our setting.

We have extended this optimization to work with partial expressions, and prune expressions in a top-down enumerator. The
idea is to replace the priority queue $Q$ with a ``deduplicating priority queue'', with the additional feature that it
removes duplicate elements from being dequeued. Such a queue can be implemented as a combination of a traditional
priority queue and a mutable set data structure commonly available in standard language libraries. Next, we maintain a
dictionary mapping all previously seen signatures to the representatives of the respective equivalence classes. Every
time we dequeue a partial production $\tilde{e}$ from $Q$, we evaluate every complete sub-production $e_{\text{sub}}$ on
all input points $\seqa \in A$ to obtain its signature $S_{\text{sub}}$, and replace $e_{\text{sub}}$ with the canonical
representative. This forms the box labelled ``Indist. rewrite 1'' in Figure~\ref{fig:alg:enum}.

We also perform a second, fast indistinguishability rewrite after each partial production is expanded. This is done by
maintaining a map from previously seen expressions to their canonical representatives. Consulting this map for every
complete sub-production of the children, $\tilde{e}_1$, $\tilde{e}_2$, \ldots, of the original production $\tilde{e}$ is
much faster than evaluating them on each concrete input point in $A$. There is some freedom in choosing the placement of
various optimizations in Figure~\ref{fig:alg:enum}: our motivation was that the priority queue can get very large, and
many elements enqueued into the queue will never be dequeued, and that it is therefore wise to postpone as much
processing as possible to when partial productions are extracted from $Q$.


\section{Term Grammars}
\label{sec:grammar-gen}

In the previous section, we explained our synthesis algorithm
using a PCFG as a black box. In this section, we explain how
the term grammars in our system are defined and extracted.

We designed our grammars with the following goals in mind:

\begin{itemize}
    \item \emph{Flexibility.} The system should allow the user to define
        and easily deploy a different grammar for each application
        or application domain.
    \item \emph{Expressivity.} The grammars should be able to express complex
        relations between expressions, in a more elaborate way than simply
        mapping each type of the target language to a set of production rules.
        This allows the user to constrain the space of synthesized programs
        to exclude redundant or irrelevant ones,
        thus improving the scalability of the synthesizer.
    \item \emph{Automation.} Apart from hand-coding the grammars,
        it should also be possible to automatically extract them 
        from a corpus of programs, possibly taken from a specific application domain.
\end{itemize}

To this end, we introduce \textit{Probabilistic Attribute Grammars},
which address all of the above points:

\begin{itemize}
    \item A custom grammar is provided by the user along with the synthesis problem
        as a set of Scala source files (the \emph{grammar files}),
        in which production rules are expressed as plain Scala functions.
        Switching between grammars reduces to just including different files in compilation.
        It is also trivial to use any desired combination of such files.
        If the user does not want to manually provide a specific grammar for the task,
        they can always fall back to the system's predefined grammar.
    \item The nonterminal symbols of the grammar are not plain types,
        but instead are types contained in a wrapper class that annotates them with extra information.
        The wrapped types can represent a subtype of the underlying type,
        or some other condition under which specific rules are accessible,
        which would not be expressible if we used plain nonterminals as types.
        While processing the input grammar, our system will desugar
        away those wrappers to produce values of the underlying plain type.

        The approach of annotating types with attributes has already been introduced
        in \cite{LeonSynth-SYNT}, but there the attributes and their interpretation
        were predefined. Here we allow the user to provide their own wrapper
        classes.
        
        It is of course also possible to omit the wrappers completely
        and use plain Scala types as nonterminals.
    \item We provide a system to automatically extract grammar files
        from a corpus of PureScala programs.
\end{itemize}

The grammar files undergo a series of preprocessing steps before generating
a Probabilistic Attribute Grammar, the exact form of which is
conditional to the synthesis problem. The synthesizer itself is oblivious
to this procedure, which views the grammar just as a function from
nonterminal symbols to production rules.

\subsection{Grammar files, rules and frequencies}
\label{grammarfiles}

As mentioned above, the user can provide a series of grammar files for compilation 
along with the synthesis problem.
The format of the files was designed with the aim to be readable and writable by a human user,
but also straightforward to extract from a corpus of programs.

To illustrate the format though an example, let us look at a grammar file describing
a set of simple arbitrary precision integer programs:

\begin{lstlisting}
@production(10)
def plus(a: BigInt, b: BigInt): BigInt = a + b

@production(5)
def minus(a: BigInt, b: BigInt): BigInt = a - b

@production(5)
def o: BigInt = BigInt(1)

@production(10)
def z: BigInt = BigInt(0)

@production(20)
def vBigInt: BigInt = variable[BigInt]
\end{lstlisting}

A function annotated with \lstinline{@production}
is treated as a grammar production rule.
A rule of the form

\lstinline{def pFoo(a: T1, b: T2): T3 = foo(a, b)}

should be interpreted in CFG notation as 

$\textsf{T3} \Coloneqq foo(T1, T2)$

The annotation's argument indicates the absolute frequency with which
this rule is encountered. The relative frequencies will be generated
in the end of a processing procedure. In this example, though, they can be computed
straightforwardly from the absolute frequencies.

An invocation of the built-in function \lstinline{variable[T]},
for some type \lstinline{T}, indicates the absolute frequency of generating
any variable of type \lstinline{T}.
This built-in compensates for the fact that we do not know the names of
the variables accessible to the synthesizer a priori.
When they become available, this production will be instantiated
to a concrete production for each variable of the correct type,
with the probability distributed equally among these variables.
If, for example, we are synthesizing a function with two parameters
\lstinline{(x: BigInt, y: BigInt)},
the above grammar file would generate the following grammar:

\[
\begin{array}{rcll}
  \textsf{BigInt}  & \Coloneqq	& \textsf{BigInt} + \textsf{BigInt} & (p_+ = 0.2) \\
 				   & \mid       & \textsf{BigInt} - \textsf{BigInt} & (p_- = 0.1) \\
				   & \mid       & \lstinline|0| & (p_0 = 0.1) \\
                   & \mid 		& \lstinline|1| & (p_1 = 0.2) \\
                   & \mid 		& \lstinline|x| & (p_x = 0.2) \\
                   & \mid 		& \lstinline|y| & (p_y = 0.2)
\end{array}
\]

\subsection{Generic productions}

Generic grammar productions offer a concise way of
representing operations on generic data-structures. For
instance, we could describe common operations on a generic
list data-structure with a few user-provided productions:

\begin{lstlisting}
@production(..)
def single[A](a: A): List[A] = List(a, Nil)

@production(..)
def listInsert[A](a: A, l: List[A]): List[A] = List(a, l)

@production(..)
def listTail[A](l: List[A]): List[A] = l.tail

@production(..)
def listHead[A](l: List[A]): A = l.head

@production(..)
def listSize[A](l: List[A]): BigInt = l.size
\end{lstlisting}

When looking for productions of a specific ground type,
we instantiate generic productions that \emph{can}
return that type, and introduce them as normal monomorphic
productions.
For example, if we are looking for productions of
\lstinline{List[BigInt]},
we consider the first three rules above instantiated with
$T \mapsto BigInt$, and the forth with $T \mapsto
List[BigInt]$. The fifth rule is incompatible.

Because the set of types is typically infinite, generic
productions with a type variable absent from the
return type represent an infinite number of ground
productions. This is for instance the case with
\lstinline{listSize}.

Our enumeration procedure described above requires
that each non-terminal symbol has a finite number of ground
productions. To solve this conflict, we abandon the
theoretical completeness of the instantiation by only
considering a set of \emph{reasonable}
types $\mathcal{T}$ during instantiation. This set of types
is initialized with all types returned by ground productions
and then iteratively expanded by discovering new return types of
$\mathcal{T}$-instantiations of generic productions.
For example, starting with $\mathcal{T} = \{\textsf{Int}\}$,
running one iteration yields $\mathcal{T} = \{\textsf{Int},
\textsf{List[Int]}\}$ as \lstinline/single[Int](..)/ returns
a value of type \lstinline/List[Int]/.

The number of discovery iterations we perform has a big
practical impact on the number of final productions, and
thus on performance. On benchmarks with several generic
data-structures, this discovery is typically exponential.
By under-approximating the size of the smallest expression
of any discovered type in $\mathcal{T}$, we bound the
search to reasonable types. 
%
We leave to future work how to better integrate the
instantiation of generic rules within the enumeration
process.

\subsection{Annotated nonterminals}
\label{annotated}
Simple grammars like the one displayed above, where the nonterminal symbols
are plain types, are not expressive enough to handle all constraints that are obvious to a human programmer.
For example, they cannot express that an operand of an addition cannot be 0.
When used as-is, such grammars generate a large number of redundant programs,
which prevent the synthesizer from scaling to larger program sizes.

To amend this, we provide the programmer with a general mechanism
to more accurately express the restrictions of their grammar
by enhancing their nonterminals with attributes
in additional to the plain Scala type.
Nonterminals with the same type but different attributes
are considered distinct, and can have different production rules.

As an example, suppose the user wants to more accurately specify the grammar
of Section~\ref{grammarfiles} to exclude 0 from the operands of \lstinline{+},
as well as the second operand of \lstinline{-}.
This can be captured with the following attribute grammar:
\[
\begin{array}{rcll}
\textsf{BigInt}         & \Coloneqq	& \textsf{BigInt\{BI\}} & (p = 1.0) \\
\textsf{BigInt\{BI\}}   & \mid      & \textsf{BigInt\{NZ\}} & (p = 0.8) \\
                        & \mid      & \lstinline|0|         & (p = 0.2) \\
\textsf{BigInt\{NZ\}}   & \mid	    & \textsf{BigInt\{NZ\}} + \textsf{BigInt\{NZ\}} & (p_1 = 0.25) \\
                        & \mid	    & \textsf{BigInt\{BI\}} - \textsf{BigInt\{NZ\}} & (p_1 = 0.125) \\
                        & \mid 		& \lstinline|x| & (p_x = 0.5) \\
                        & \mid 		& \lstinline|1| & (p_y = 0.125)
\end{array}
\]

The above grammar contains two annotated nonterminals,
\textsf{BigInt\{NZ\}} and \textsf{BigInt\{BI\}}.
The former expresses nonzero integers and the second one
the unrestricted integer type. Thus, generating \textsf{BigInt} values
reduces to generating values of the annotated type 
\textsf{BigInt\{BI\}},
which is captured in the first rule of the grammar.

In Scala source code, this is expressed by defining
implicit Scala classes containing a field of the base type,
and annotating them with the \lstinline{@label} annotation.
The source file gets desugared by the system to the above grammar.

\begin{lstlisting}
@label implicit class NZ(val v: BigInt)

@label implicit class BI(val v: BigInt)

@production(10) def plus(a: NZ, b: NZ): NZ = a.v + b.v

@production(5)  def minus(a: BI, b: NZ): NZ = a.v - b.v

@production(5)  def o: NZ = BigInt(1)

@production(10) def z: BI = BigInt(0)

@production(20) def vBigInt: NZ = variable[BigInt]

@production(40) def nz2Bi(nz: NZ): BI = nz.v

@production(1)  def start(b: BI): BigInt = b.v

\end{lstlisting}

\subsubsection{Built-in axiom system}
\label{axioms}
In case the user does not want to manually implement a custom attribute grammar as above,
they can reside to a set of default built-in axioms described in \cite{LeonSynth-SYNT}.
In case of the default grammars, those axioms are automatically applied to each combination of
nonterminal and production rule.
In case of a manually provided grammar, the user has to guide the system by
tagging the production rules with a set of predefined tags.
A couple of example tagged rules could be the following:

\begin{lstlisting}
@production(10)
@tag("0")
def z: BigInt = BigInt(0)

@production(10)
@tag("commut")
def pFoo(b1: BigInt, b2: BigInt): BigInt = foo(b1, b2)
\end{lstlisting}

\lstinline{@tag("0")} indicates that this rule produces the constant 0
(e.g. it should be treated as neutral element of addition),
whereas \lstinline{@tag("commut")} specifies \lstinline{foo} as a commutative
operation. These rules will undergo a set of predetermined optimizations
described in \cite{LeonSynth-SYNT}.

\subsection{Extracting an attribute grammar form a corpus}

Whereas manually specifying grammars is very useful for specific tasks and small DSLs,
it would be tedious to expect a user-provided grammar for generic programming tasks,
where numerous types and functions are in scope.

To this end, we provide an automated system to extract a grammar from a corpus
of Scala programs. We provide two different grammar extractors,
described in the following paragraphs.

\subsubsection{Unconditional Frequency Extractor}

The simpler Unconditional Frequency Extractor (UFE) uses plain Scala types as nonterminals,
but automatically tags them as in \ref{axioms},
so the built-in axioms can be instantiated on them.
Each production rule generated by the UFE corresponds to a specific \emph{kind}
of expression encountered in the corpus,
whereas the rule's frequency reflects the number of occurrences
of this kind.
By kind of expression, we mean the expression type
along with all information required to reconstruct the expression
from its arguments (except for variables, where the name is not preserved).
More specifically,
\begin{itemize}
    \item For variables, only the variable type (the name does not get preserved).
        All variables of a specific type will all be summarized by an invocation
        to the \lstinline{variable} built-in.
    \item For literals, the literal itself (which takes 0 operands).
        E.g. \lstinline{0}, \lstinline{BigInt(42)}.
    \item For function invocations, the invoked function along with its actual type parameters
        (which may be generic).
        E.g. \lstinline{foo[BigInt](?, ?)}, \lstinline{foo[T](?, ?)},
		where by \lstinline{?} we denote a hole to be filled during synthesis.
    \item For type constructors, the type constructor along with its actual
        type parameters. E.g. \lstinline{Cons[BigInt](?, ?)},
        \lstinline{Nil[A]()}.
    \item For every other expression, only the type of the AST of that expression,
        and if needed its type instantiation.
        E.g. \lstinline{? + ?}, \lstinline{Set[BigInt](?)}.
\end{itemize}

Because the grammars generated by this extractor do not take into
account the parent of the expression,
we will call them \emph{depth-1 grammars}.

For example, a corpus containing only the expression \lstinline{(x * 2)}
would result in the following file:

\begin{lstlisting}
@production(1)
@tag("const")
def pBigIntInfiniteIntegerLiteral0(): BigInt = BigInt(2)

@production(1)
@tag("times")
def pBigIntTimes(v0 : BigInt, v1 : BigInt): BigInt = v0 * v1

@production(1)
@tag("top")
def pBigIntVariable(): BigInt = variable[BigInt]
\end{lstlisting}

\subsubsection{Conditional Frequency Extractor}

The second, more sophisticated extractor, extracts the conditional
probabilities of the occurrence of each expression kind,
given the expression's relation with its parent.
More specifically, along with each expression, it stores a pair
$(par, pos)$, where $par$ is the type of AST of the parent of the expression,
and $pos$ the expression's position in its parent's operands.
This pair is optional and omitted for a top-level expression.
For example, if the extractor encountered the expression \lstinline{(x * 2)},
it would store the triples \lstinline{(? * ?, null, null)},
\lstinline{(<variable>, ? * ?, 0)} and \lstinline{(2, ? * ?, 1)}.
To express these conditional probabilities in a grammar file,
it uses annotated nonterminals as described in \ref{annotated}.

Because these grammars express probabilities which are derived from
an expression along with its parent,
we will call them \emph{depth-2 grammars}.

A corpus containing only the above expression would produce the following
grammar file:

\begin{lstlisting}
@label implicit class BigInt_TOPLEVEL(val v : BigInt)  
@label implicit class BigInt_0_Times(val v : BigInt)  
@label implicit class BigInt_1_Times(val v : BigInt)  

@production(1)
def pBigIntTimes(v0 : BigInt_0_Times, v1 : BigInt_1_Times): BigInt_TOPLEVEL =
  v0.v * v1.v  
@production(1)
def pBigIntVariable(): BigInt_0_Times = variable[BigInt]  
@production(1)
def pBigIntInfiniteIntegerLiteral(): BigInt_1_Times = BigInt(2)  
@production(1)
def pBigIntStart(v0 : BigInt_TOPLEVEL): BigInt = v0.v
\end{lstlisting}

\subsection{Biasing grammars for repair}

The flexibility and automation of the system give us a great opportunity
to exploit the local structure of the program.
Especially for program repair, we can ease the task of generating a reasonable
patch for an erroneous program by collecting information about how 
the user has chosen to program in the same module.

This is achieved as follows: before repair, we generate an additional grammar file
from the file under repair, and include it for compilation along with
the general-purpose grammar file extracted from the corpus.
This way, we bias the synthesis process towards local definitions,
without completely disregarding the general programming model.

\section{Experimental Results}
\label{sec:experiments}

To experimentally evaluate the above ideas, we integrated them into
the Leon synthesis and repair framework \cite{LeonSynth-OOPSLA}
The Probabilistic Enumeration (PE) of Section~\ref{sec:alg} is deployed as
a closing deductive synthesis rule in the synthesis framework,
in place of Symbolic Term Exploration (STE) \cite{LeonSynth-SYNT}.

For repair, we slightly modify Probabilistic Enumeration to
generate terms similar to the erroneous expression, in the spirit of 
\cite{LeonSynth-CAVRepair}.
This can be achieved by automatically modifying the grammar
passed to Enumeration, without any other changes to the algorithm.
This rule is deployed along with plain PE within the synthesis stage
of the repair pipeline,
after test generation and fault localization

In our evaluation, we aimed to measure how well our algorithms
perform in isolation, rather than how well they integrate with the
synthesis framework.
This is straightforward in repair,
as the erroneous implementation already defines the
high-level structure of the function under repair,
and the closing rules can usually be deployed alone,
without instantiating any decomposition rules first.

In the following tables, we compare the efficiency of synthesizing repairs
for a set of erroneous benchmarks taken from \cite{LeonSynth-SYNT},
along with a few new, more challenging ones.

In total, we test 4 tool configurations:

\begin{itemize}
    \item Symbolic Term Exploration with built-in grammars and 
        built-in axioms, i.e. the old version of Leon (\emph{STE})
    \item Probabilistic Enumeration with built-in grammars
        and built-in axioms (\emph{PB})
    \item Probabilistic Enumeration with depth-1 extracted grammars
        and built in axioms (\emph{PD1})
    \item Probabilistic Enumeration with depth-2 grammars,
        without built-in axioms (\emph{PD2})
\end{itemize}

In both the latter cases, the grammars come from a combination
of grammar files extracted from a corpus and the program under repair.
We do not use built-in axioms with depth-2 grammars,
as they are redundant to some degree and they would overcomplicate the
grammar, while yielding diminishing returns.

The detailed results are given in \ref{table:repair}.
The first two columns of the table give an indication of the complexity of the problem
by displaying the total size of the program and the size of the solution
(which might not be generated from scratch, but rather as a variation
of the erroneous expression). The final four columns display
the repair times for the four deployments of the tool described above.
Repair times exclude compilation, test generation, and verification of the solution.

\newcommand{\leg}[2]{{#1: #2\newline }}
\newcommand{\mc}[2]{\multicolumn{2}{#1}{#2}}
\newcommand{\tme}{time}
\renewcommand{\bname}[1]{{\footnotesize \texttt{#1}}}
\newcolumntype{R}{>{\hspace{-7pt}}r}
\newcolumntype{C}{>{\hspace{-3pt}}c}

\begin{table}
\begin{center}
\begin{tabular}{l|rr|r|r|r|r}
Operation                     & \mc{c|}{Sizes} &  STE     &   PB    &  PD1  &    PD2  \\
                              & Prog &  Sol    &          &         &       &         \\
\hline                                                                                 
\bname{ Compiler.desugar    } & 717  &     3   &   3.3    &   1.9   &  2.3  &    2.9   \\
\bname{ Compiler.desugar    } & 715  &     2   &   5.2    &   3.7   &  4.3  &    3.4   \\
\bname{ Compiler.desugar    } & 719  &     7   &   3.0    &   1.8   &  2.1  &    2.0   \\
\bname{ Compiler.desugar    } & 719  &     7   &   3.8    &   3.0   &  3.0  &    3.2   \\
\bname{ Compiler.desugar    } & 719  &    14   &   5.2    &   3.4   &  3.2  &    3.7   \\
\bname{ Compiler.simplify   } & 765  &     4   &   2.8    &   1.8   &  1.8  &    3.1   \\
\bname{ Compiler.semUntyped } & 738  &     5   &          &         & 16.8  &    5.1   \\
\bname{ Heap.merge          } & 382  &     3   &   6.0    &   4.9   &  5.1  &    5.2   \\
\bname{ Heap.merge          } & 382  &     1   &   2.8    &   2.0   &  2.0  &    2.2   \\
\bname{ Heap.merge          } & 382  &     3   &   5.7    &   4.7   &  5.1  &    5.3   \\
\bname{ Heap.merge          } & 382  &     9   &   4.4    &   3.6   &  3.4  &    3.7   \\
\bname{ Heap.merge          } & 384  &     5   &   5.1    &   6.4   &  5.2  &    15.6  \\
\bname{ Heap.merge          } & 382  &     2   &   4.2    &   2.6   &  2.7  &    2.9   \\
\bname{ Heap.insert         } & 345  &     8   &   3.0    &   2.8   &  2.9  &    2.9   \\
\bname{ Heap.makeN          } & 384  &     7   &   5.6    &   7.9   &  8.7  &    8.6   \\
\bname{ List.pad            } & 830  &     8   &   2.3    &   1.5   &  1.6  &    2.7   \\
\bname{ List.++             } & 740  &     3   &   2.6    &   1.5   &  1.2  &    8.1   \\
\bname{ List.:+             } & 772  &     1   &   3.0    &   1.1   &  1.2  &    1.9   \\
\bname{ List.replace        } & 774  &     6   &   3.3    &   1.7   &  1.7  &    1.7   \\
\bname{ List.count          } & 827  &     3   &   2.2    &   1.2   &  1.4  &    1.3   \\
\bname{ List.find           } & 827  &     2   &   2.5    &   1.7   &  1.6  &    1.7   \\
\bname{ List.find           } & 829  &     4   &   2.8    &   1.6   &  1.8  &    1.7   \\
\bname{ List.find           } & 830  &     4   &   2.7    &   1.7   &  1.9  &    1.7   \\
\bname{ List.size           } & 755  &     4   &   2.5    &   1.1   &       &          \\
\bname{ List.sum            } & 774  &     4   &   2.3    &   1.5   &  1.5  &    1.5   \\
\bname{ List.-              } & 774  &     1   &   1.5    &   0.5   &  0.5  &    1.7   \\
\bname{ List.drop           } & 815  &     4   &   2.4    &   1.4   &  1.6  &    1.5   \\
\bname{ List.drop           } & 815  &     3   &   4.0    &   3.3   & 46.9  &    2.4   \\
\bname{ List.\&             } & 774  &     4   &          &         &  3.1  &          \\
\bname{ List.count          } & 774  &     1   &          &         &170.7  &          \\
\bname{ Numerical.power     } & 225  &     5   &   2.9    &   1.1   &  1.2  &          \\
\bname{ Numerical.moddiv    } & 174  &     3   &   1.9    &   0.8   &  0.8  &    1.0   \\
\bname{ MergeSort.split     } & 284  &     5   &   3.6    &   2.3   &  2.5  &    8.4   \\
\bname{ MergeSort.merge     } & 286  &     7   &   3.3    &   2.3   &  2.5  &    2.2   \\
\bname{ MergeSort.merge     } & 286  &     3   &   5.9    &   4.8   &  5.2  &    5.0   \\
\bname{ MergeSort.merge     } & 284  &     5   &   3.3    &   2.4   &  2.3  &    13.7  \\
\bname{ MergeSort.merge     } & 286  &     1   &   2.3    &   1.6   &  1.7  &    1.7   \\
\end{tabular}
\caption{Benchmarks for repair \label{table:repair}}
\end{center}
\end{table}

To summarize the results, the new enumerator is slightly faster for the built-in grammars,
whereas the depth-1 grammars are sometimes slightly slower, but manage to solve a couple more benchmarks.
Unfortunately, there is a relatively easy benchmark that depth-1 fails to solve.
The depth-2 grammars, though they show some promise (in the
\lstinline{Compiler.semUntyped} benchmark), they generally perform badly.
One reason for this behavior are the conflicts between annotated nonterminal symbols
generated from the general corpus with those from the file under repair.
Both sources will generate annotated non-terminals that express the same conditional
probability, overloading the synthesizer with duplicated work.
This is a shortcoming we plan to address in the near future.

\section{Related Work}
\label{sec:related}

\subsection{Program synthesis}
\label{sec:related:synt}

Program synthesis has been a challenge problem in computer science since its early days~(see, for example,~%
\cite{Manna-1971}). With increasing computing power, and rapid improvements in the performance of SAT and SMT solvers
starting from the early 2000s, the problem of program synthesis has recently received renewed attention. The seminal
Sketch project~\cite{Sketch} was the first to demonstrate the feasibility of using modern constraint solving technology
in synthesis tools, and used the popular counter-example guided inductive synthesis~(CEGIS) framework to solve synthesis
problems~\cite{Gulwani-etal-2011}. 

Syntax-guided synthesis~(SyGuS) has emerged as a common formulation and interchange format in which to express many
program synthesis problems~\cite{SyGuS}. Our problem formulation within Leon shares two important similarities with
SyGuS:
\begin{enumerate}
\item \emph{Multi-modal specifications}, which can include any combination of input-output examples~(such as ``$f(2, 3)
  = 8$''), logical constraints~(such as ``$\forall x, y, f(x, y) > x + y + 2$''), and even partial implementations~(such
  as ``$\forall x, y, x < y \implies f(x, y) = x + y + 3$'').
\item \emph{Expression grammars}, or equivalently, a library of \emph{components}~\cite{Brahma}, are used to
  syntactically constrain the space of programs. This is important both to make synthesis feasible for complicated
  specifications, and for when the results are to be used in specialized domains, such as assembly code for new computer
  architectures~\cite{Chlorophyll}.
\end{enumerate}
Of the various techniques to solve SyGuS problems, ESolver, which enumerates pairwise distinguishable expressions~%
\cite{Transit}, EUSolver, which combines indistinguishability-based enumeration with decision-tree based condition
inference~\cite{EUSolver}, and the refutation-proof-based CVC4 solver~\cite{CVC4-SyGuS} have broadly emerged as the most
efficient SyGuS solvers. Our present work may be viewed as an attempt to bias the enumerative solver with statistical
facts learned from a code corpus.

We would like to highlight that program synthesis and code repair within Leon is a much more challenging problem than
synthesis within SyGuS and related systems: the main difficulties include the rich type system and algebraic data types
of Scala, and the synthesis of recursive functions. In this light, our problem is most similar to
Escher~\cite{Albarghouthi2013} and $\lambda^2$~\cite{Feser2015}: the main differences are that Escher enumerates
programs bottom-up rather than top-down, and our use of PCFGs to improve the performance of the enumeration algorithm.


\subsection{Type-driven synthesis}
\label{sec:related:types}

One appealing approach to program synthesis is to phrase the problem as the type-inhabitation problem in a sufficiently
expressive type system. Examples of such systems include InSynth~\cite{InSynth}, and Synquid~\cite{Synquid}. The Myth
project~\cite{Osera-Myth-15, Osera-Myth-16} extends the type system of the host language with refinement types
corresponding to input-output examples. It can therefore be viewed as an attempt to provide a type-theoretic
interpretation to Escher. The Synquid system~\cite{Synquid} is an elegant extension of this idea, by encoding the
specification as part of the type of the required expression using the liquid type framework.

Previous work on synthesis within Leon~\cite{LeonSynth-OOPSLA, LeonSynth-SYNT} may also be viewed as being implicitly
type-driven. The path conditions and postcondition obtained after each invocation of the deduction rules can be
rephrased as additional typing constraints on the expression being synthesized. As discussed in Section~\ref{sec:alg},
the present PCFG-based enumeration is intended as a drop-in substitute for the symbolic term exploration rule previously
used by Leon.


\subsection{Statistical analysis of code corpora}
\label{sec:related:corpora}

With increasing availability of large open-source code repositories such as GitHub and Bitbucket, the statistical
analysis of code corpora has become an exciting research problem. Code repositories have been used to learn coding
idioms~\cite{AllamanisSutton-Idioms}, to automatically suggest names for program elements~\cite{AllamanisSutton-Names},
and to deobfuscate JavaScript code~\cite{Vechev-BigCode}. Our paper is an attempt to apply similar techniques to
accelerate program synthesis.

The program completion tool Slang~\cite{Vechev-PLDI14} uses $n$-grams and recurrent neural networks to predict missing
API calls in code snippets. There are two main aspects which distinguish our work from Slang:
\begin{inparaenum}[(\itshape a\upshape)]
\item the presence of hard correctness requirements in Leon in the form of pre-/post-conditions, and
\item program synthesis in Leon is fundamentally about synthesizing expressions rather than API call sequences, and
  prediction systems such as $n$-grams are insufficient to create the nested recursive structure inherent in the output
  we produce.
\end{inparaenum}
The DeepCoder tool~\cite{DeepCoder} uses a recurrent neural network (RNN) to predict the presence of elements in the
program being synthesized. The output of this neural network is then used to guide a more exhaustive search over the
space of possible programs. In this sense, DeepCoder is similar to the PCFG-guided enumeration of Leon, and the use of
neural networks to accelerate the program synthesis problem is an intriguing direction for future research.

Probabilistic context-free grammars (PCFGs) are a classical extension of context-free grammars~\cite{NLP-Jurafsky}. They
may be used both to model ambiguity (for applications in natural language processing), and to model probability
distributions over the generated language, which motivates our application in accelerating code synthesis. The more
recent model of probabilistic higher-order grammars (PHOGs)~\cite{Vechev-PHOG, Vechev-PHOG0} extends PCFGs by allowing
the expansion probabilities of a non-terminal node to depend on attributes such as node siblings and DFS-predecessors.
Experiments indicate that the PHOG model is significantly better at predicting elements of JavaScript programs than
PCFGs. Extending the probabilistic model of our paper to use PHOGs instead of PCFGs is an immediate area of future work.

\subsection{Search algorithms}
\label{sec:related:search}

One of the main technical insights of this paper is to phrase syntax-guided program synthesis as an instance of graph
search. In this space, A*~\cite{AStar} is a classical algorithm used for path finding and graph traversal. As described
in Section~\ref{sec:alg}, A* orders the queue of still-unexplored nodes by their priority values, $f(n) = g(n) + h(n)$,
where $g(n)$ is the cost required to reach the node $n$, and the heuristic function $h(n)$ is the estimated cost from
$n$ to the target node. The historically older Dijkstra's algorithm for shortest paths through graphs is a special case,
where the heuristic function $h(n)$ is identically equal to $0$.

Recall that the graph of partial expansions forms a tree: in this setting, if the cost function $h(n)$ is
\emph{admissible}, i.e. that $h(n) \leq \tilde{h}(n)$, where $\tilde{h}(n)$ is the actual cost of the shortest path from
$n$ to the target node, then A* discovers the shortest path to the target node (\emph{optimality}). The second important
property of A* is that of \emph{optimal efficiency}: of all algorithms using the same heuristic function, A* explores
the fewest nodes to reach the target. In the setting of this paper, optimality implies discovering the expansion with
highest probability. By extending $h(n)$ to include the number of satisfied examples, we have relaxed optimality in
favour of empirical performance improvements.

A* is simply the best known of a large family of search algorithms~(see chapters~3 and~4 of \cite{AIMA}), including
iterative deepening A* (IDA*,~\cite{IDAStar}), simplified memory-bounded A* (SMA*,~\cite{SMAStar}), fringe
search~\cite{FringeSearch}, beam search, swarm search, and various instantiations of genetic
algorithms~\cite{GAMitchell}. Extending the techniques of this paper to these algorithms is an exciting direction of
future work.

\section{Conclusion + Future Work}
\label{sec:conclusion}

In this paper, we present a technique to use statistical information extracted from code corpora, encoded as
probabilistic attribute grammars, to guide program synthesis techniques. The method is based on rephrasing expression
enumeration as graph search, and applying classical techniques such as A* to solve these problems. We implemented the
algorithm in the Leon synthesis and verification system, and obtained a much more robust and scalable synthesizer in
comparison to the previous implementation. Our end goal is a predictable one-touch program synthesizer in a rich
programming environment such as Scala: there are several exciting directions of future work, such as investigating other
search methods including swarm search and genetic algorithms, incorporating richer models of statistics such as
probabilistic higher-order grammars (PHOGS)~\cite{Vechev-PHOG}, and combining type- and syntax-driven and statistical
techniques into a highly expressive and performant program synthesizer.

\bibliographystyle{eptcs}
\bibliography{references}

\begin{thebibliography}{10}
\providecommand{\bibitemdeclare}[2]{}
\providecommand{\surnamestart}{}
\providecommand{\surnameend}{}
\providecommand{\urlprefix}{Available at }
\providecommand{\url}[1]{\texttt{#1}}
\providecommand{\href}[2]{\texttt{#2}}
\providecommand{\urlalt}[2]{\href{#1}{#2}}
\providecommand{\doi}[1]{doi:\urlalt{http://dx.doi.org/#1}{#1}}
\providecommand{\bibinfo}[2]{#2}

\bibitemdeclare{inproceedings}{Albarghouthi2013}
\bibitem{Albarghouthi2013}
\bibinfo{author}{Aws \surnamestart Albarghouthi\surnameend},
  \bibinfo{author}{Sumit \surnamestart Gulwani\surnameend} \&
  \bibinfo{author}{Zachary \surnamestart Kincaid\surnameend}
  (\bibinfo{year}{2013}): \emph{\bibinfo{title}{Recursive Program Synthesis}}.
\newblock In \bibinfo{editor}{Natasha \surnamestart Sharyngina\surnameend} \&
  \bibinfo{editor}{Helmut \surnamestart Veith\surnameend}, editors: {\sl
  \bibinfo{booktitle}{25th International Conference on Computer Aided
  Verification}}, \bibinfo{publisher}{Springer}, pp. \bibinfo{pages}{934--950},
  \doi{10.1007/978-3-642-39799-8}.
\newblock \urlprefix\url{http://dx.doi.org/10.1007/978-3-642-39799-8_67}.

\bibitemdeclare{inproceedings}{AllamanisSutton-Names}
\bibitem{AllamanisSutton-Names}
\bibinfo{author}{Miltiadis \surnamestart Allamanis\surnameend},
  \bibinfo{author}{Earl \surnamestart Barr\surnameend},
  \bibinfo{author}{Christian \surnamestart Bird\surnameend} \&
  \bibinfo{author}{Charles \surnamestart Sutton\surnameend}
  (\bibinfo{year}{2015}): \emph{\bibinfo{title}{Suggesting Accurate Method and
  Class Names}}.
\newblock In: {\sl \bibinfo{booktitle}{Proceedings of the 10th Joint Meeting on
  Foundations of Software Engineering}}, \bibinfo{series}{ESEC/FSE 2015},
  \bibinfo{publisher}{ACM}, pp. \bibinfo{pages}{38--49},
  \doi{10.1145/2786805.2786849}.
\newblock \urlprefix\url{http://doi.acm.org/10.1145/2786805.2786849}.

\bibitemdeclare{inproceedings}{AllamanisSutton-Idioms}
\bibitem{AllamanisSutton-Idioms}
\bibinfo{author}{Miltiadis \surnamestart Allamanis\surnameend} \&
  \bibinfo{author}{Charles \surnamestart Sutton\surnameend}
  (\bibinfo{year}{2014}): \emph{\bibinfo{title}{Mining Idioms from Source
  Code}}.
\newblock In: {\sl \bibinfo{booktitle}{Proceedings of the 22nd ACM SIGSOFT
  International Symposium on Foundations of Software Engineering}},
  \bibinfo{series}{FSE 2014}, \bibinfo{publisher}{ACM}, pp.
  \bibinfo{pages}{472--483}, \doi{10.1145/2635868.2635901}.
\newblock \urlprefix\url{http://doi.acm.org/10.1145/2635868.2635901}.

\bibitemdeclare{inproceedings}{SyGuS}
\bibitem{SyGuS}
\bibinfo{author}{Rajeev \surnamestart Alur\surnameend},
  \bibinfo{author}{Rastislav \surnamestart Bodik\surnameend},
  \bibinfo{author}{Garvit \surnamestart Juniwal\surnameend},
  \bibinfo{author}{Milo \surnamestart Martin\surnameend},
  \bibinfo{author}{Mukund \surnamestart Raghothaman\surnameend},
  \bibinfo{author}{Sanjit \surnamestart Seshia\surnameend},
  \bibinfo{author}{Rishabh \surnamestart Singh\surnameend},
  \bibinfo{author}{Armando \surnamestart Solar-Lezama\surnameend},
  \bibinfo{author}{Emina \surnamestart Torlak\surnameend} \&
  \bibinfo{author}{Abhishek \surnamestart Udupa\surnameend}
  (\bibinfo{year}{2013}): \emph{\bibinfo{title}{Syntax-guided synthesis}}.
\newblock In: {\sl \bibinfo{booktitle}{2013 Formal Methods in Computer-Aided
  Design}}, pp. \bibinfo{pages}{1--8}, \doi{10.1109/FMCAD.2013.6679385}.

\bibitemdeclare{inproceedings}{EUSolver}
\bibitem{EUSolver}
\bibinfo{author}{Rajeev \surnamestart Alur\surnameend}, \bibinfo{author}{Arjun
  \surnamestart Radhakrishna\surnameend} \& \bibinfo{author}{Abhishek
  \surnamestart Udupa\surnameend} (\bibinfo{year}{2017}):
  \emph{\bibinfo{title}{Scaling Enumerative Program Synthesis via Divide and
  Conquer}}.
\newblock In \bibinfo{editor}{Axel \surnamestart Legay\surnameend} \&
  \bibinfo{editor}{Tiziana \surnamestart Margaria\surnameend}, editors: {\sl
  \bibinfo{booktitle}{23rd International Conference on Tools and Algorithms for
  the Construction and Analysis of Systems}}, \bibinfo{publisher}{Springer},
  pp. \bibinfo{pages}{319--336}, \doi{10.1007/978-3-662-54577-5}.
\newblock \urlprefix\url{http://dx.doi.org/10.1007/978-3-662-54577-5}.

\bibitemdeclare{inproceedings}{DeepCoder}
\bibitem{DeepCoder}
\bibinfo{author}{Matej \surnamestart Balog\surnameend},
  \bibinfo{author}{Alexander \surnamestart Gaunt\surnameend},
  \bibinfo{author}{Marc \surnamestart Brockschmidt\surnameend},
  \bibinfo{author}{Sebastian \surnamestart Nowozin\surnameend} \&
  \bibinfo{author}{Daniel \surnamestart Tarlow\surnameend}
  (\bibinfo{year}{2017}): \emph{\bibinfo{title}{{DeepCoder}: {L}earning to
  Write Programs}}.
\newblock In: {\sl \bibinfo{booktitle}{International Conference on Learning
  Representations}}, \bibinfo{volume}{abs/1611.01989}.
\newblock \urlprefix\url{http://arxiv.org/abs/1611.01989}.

\bibitemdeclare{inproceedings}{Vechev-PHOG}
\bibitem{Vechev-PHOG}
\bibinfo{author}{Pavol \surnamestart Bielik\surnameend},
  \bibinfo{author}{Veselin \surnamestart Raychev\surnameend} \&
  \bibinfo{author}{Martin \surnamestart Vechev\surnameend}
  (\bibinfo{year}{2016}): \emph{\bibinfo{title}{{PHOG}: {P}robabilistic Model
  for Code}}.
\newblock In \bibinfo{editor}{Maria~Florina \surnamestart Balcan\surnameend} \&
  \bibinfo{editor}{Kilian~Q. \surnamestart Weinberger\surnameend}, editors:
  {\sl \bibinfo{booktitle}{Proceedings of The 33rd International Conference on
  Machine Learning}}, {\sl \bibinfo{series}{Proceedings of Machine Learning
  Research}}~\bibinfo{volume}{48}, \bibinfo{publisher}{PMLR}, pp.
  \bibinfo{pages}{2933--2942}.
\newblock \urlprefix\url{http://proceedings.mlr.press/v48/bielik16.html}.

\bibitemdeclare{inproceedings}{FringeSearch}
\bibitem{FringeSearch}
\bibinfo{author}{Yngvi \surnamestart Bj{\"{o}}rnsson\surnameend},
  \bibinfo{author}{Markus \surnamestart Enzenberger\surnameend},
  \bibinfo{author}{Robert \surnamestart Holte\surnameend} \&
  \bibinfo{author}{Jonathan \surnamestart Schaeffer\surnameend}
  (\bibinfo{year}{2005}): \emph{\bibinfo{title}{Fringe Search: {B}eating {A*}
  at Pathfinding on Game Maps}}.
\newblock In: {\sl \bibinfo{booktitle}{Proceedings of the 2005 {IEEE} Symposium
  on Computational Intelligence and Games}}.
\newblock \urlprefix\url{http://csapps.essex.ac.uk/cig/2005/papers/p1039.pdf}.

\bibitemdeclare{inproceedings}{Leon-Scala13}
\bibitem{Leon-Scala13}
\bibinfo{author}{R{\'e}gis \surnamestart Blanc\surnameend},
  \bibinfo{author}{Viktor \surnamestart Kuncak\surnameend},
  \bibinfo{author}{Etienne \surnamestart Kneuss\surnameend} \&
  \bibinfo{author}{Philippe \surnamestart Suter\surnameend}
  (\bibinfo{year}{2013}): \emph{\bibinfo{title}{An Overview of the {L}eon
  Verification System: {V}erification by Translation to Recursive Functions}}.
\newblock In: {\sl \bibinfo{booktitle}{Proceedings of the 4th Workshop on
  Scala}}, \bibinfo{series}{SCALA '13}, \bibinfo{publisher}{ACM}, pp.
  \bibinfo{pages}{1:1--1:10}, \doi{10.1145/2489837.2489838}.
\newblock \urlprefix\url{http://doi.acm.org/10.1145/2489837.2489838}.

\bibitemdeclare{inproceedings}{Feser2015}
\bibitem{Feser2015}
\bibinfo{author}{John \surnamestart Feser\surnameend}, \bibinfo{author}{Swarat
  \surnamestart Chaudhuri\surnameend} \& \bibinfo{author}{Isil \surnamestart
  Dillig\surnameend} (\bibinfo{year}{2015}): \emph{\bibinfo{title}{Synthesizing
  Data Structure Transformations from Input-output Examples}}.
\newblock In: {\sl \bibinfo{booktitle}{Proceedings of the 36th ACM SIGPLAN
  Conference on Programming Language Design and Implementation}},
  \bibinfo{series}{PLDI '15}, \bibinfo{publisher}{ACM}, pp.
  \bibinfo{pages}{229--239}, \doi{10.1145/2737924.2737977}.
\newblock \urlprefix\url{http://doi.acm.org/10.1145/2737924.2737977}.

\bibitemdeclare{inproceedings}{Osera-Myth-16}
\bibitem{Osera-Myth-16}
\bibinfo{author}{Jonathan \surnamestart Frankle\surnameend},
  \bibinfo{author}{Peter-Michael \surnamestart Osera\surnameend},
  \bibinfo{author}{David \surnamestart Walker\surnameend} \&
  \bibinfo{author}{Steve \surnamestart Zdancewic\surnameend}
  (\bibinfo{year}{2016}): \emph{\bibinfo{title}{Example-directed Synthesis: {A}
  Type-theoretic Interpretation}}.
\newblock In: {\sl \bibinfo{booktitle}{Proceedings of the 43rd Annual ACM
  SIGPLAN-SIGACT Symposium on Principles of Programming Languages}},
  \bibinfo{series}{POPL '16}, \bibinfo{publisher}{ACM}, pp.
  \bibinfo{pages}{802--815}, \doi{10.1145/2837614.2837629}.
\newblock \urlprefix\url{http://doi.acm.org/10.1145/2837614.2837629}.

\bibitemdeclare{inproceedings}{Gulwani-etal-2011}
\bibitem{Gulwani-etal-2011}
\bibinfo{author}{Sumit \surnamestart Gulwani\surnameend},
  \bibinfo{author}{Susmit \surnamestart Jha\surnameend},
  \bibinfo{author}{Ashish \surnamestart Tiwari\surnameend} \&
  \bibinfo{author}{Ramarathnam \surnamestart Venkatesan\surnameend}
  (\bibinfo{year}{2011}): \emph{\bibinfo{title}{Synthesis of Loop-free
  Programs}}.
\newblock In: {\sl \bibinfo{booktitle}{Proceedings of the 32nd ACM SIGPLAN
  Conference on Programming Language Design and Implementation}},
  \bibinfo{series}{PLDI '11}, \bibinfo{publisher}{ACM}, pp.
  \bibinfo{pages}{62--73}, \doi{10.1145/1993498.1993506}.
\newblock \urlprefix\url{http://doi.acm.org/10.1145/1993498.1993506}.

\bibitemdeclare{inproceedings}{InSynth}
\bibitem{InSynth}
\bibinfo{author}{Tihomir \surnamestart Gvero\surnameend},
  \bibinfo{author}{Viktor \surnamestart Kuncak\surnameend},
  \bibinfo{author}{Ivan \surnamestart Kuraj\surnameend} \&
  \bibinfo{author}{Ruzica \surnamestart Piskac\surnameend}
  (\bibinfo{year}{2013}): \emph{\bibinfo{title}{Complete Completion Using Types
  and Weights}}.
\newblock In: {\sl \bibinfo{booktitle}{Proceedings of the 34th ACM SIGPLAN
  Conference on Programming Language Design and Implementation}},
  \bibinfo{series}{PLDI 2013}, \bibinfo{publisher}{ACM}, pp.
  \bibinfo{pages}{27--38}, \doi{10.1145/2491956.2462192}.
\newblock \urlprefix\url{http://doi.acm.org/10.1145/2491956.2462192}.

\bibitemdeclare{article}{AStar}
\bibitem{AStar}
\bibinfo{author}{Peter \surnamestart Hart\surnameend}, \bibinfo{author}{Nils
  \surnamestart Nilsson\surnameend} \& \bibinfo{author}{Bertram \surnamestart
  Raphael\surnameend} (\bibinfo{year}{1968}): \emph{\bibinfo{title}{A formal
  basis for the heuristic determination of minimum cost paths}}.
\newblock {\sl \bibinfo{journal}{IEEE Transactions on Systems, Science, and
  Cybernetics}} \bibinfo{volume}{SSC-4}(\bibinfo{number}{2}), pp.
  \bibinfo{pages}{100--107}.

\bibitemdeclare{inproceedings}{Brahma}
\bibitem{Brahma}
\bibinfo{author}{Susmit \surnamestart Jha\surnameend}, \bibinfo{author}{Sumit
  \surnamestart Gulwani\surnameend}, \bibinfo{author}{Sanjit \surnamestart
  Seshia\surnameend} \& \bibinfo{author}{Ashish \surnamestart
  Tiwari\surnameend} (\bibinfo{year}{2010}):
  \emph{\bibinfo{title}{Oracle-guided Component-based Program Synthesis}}.
\newblock In: {\sl \bibinfo{booktitle}{Proceedings of the 32nd ACM / IEEE
  International Conference on Software Engineering - Volume 1}},
  \bibinfo{series}{ICSE '10}, \bibinfo{publisher}{ACM}, pp.
  \bibinfo{pages}{215--224}, \doi{10.1145/1806799.1806833}.
\newblock \urlprefix\url{http://doi.acm.org/10.1145/1806799.1806833}.

\bibitemdeclare{book}{NLP-Jurafsky}
\bibitem{NLP-Jurafsky}
\bibinfo{author}{Daniel \surnamestart Jurafsky\surnameend} \&
  \bibinfo{author}{James \surnamestart Martin\surnameend}
  (\bibinfo{year}{2008}): \emph{\bibinfo{title}{Speech and Language
  Processing}}, \bibinfo{edition}{2nd} edition.
\newblock \bibinfo{publisher}{Prentice Hall}.
\newblock \bibinfo{note}{See chapter~14 on statistical parsing}.

\bibitemdeclare{inproceedings}{LeonSynth-CAVRepair}
\bibitem{LeonSynth-CAVRepair}
\bibinfo{author}{Etienne \surnamestart Kneuss\surnameend},
  \bibinfo{author}{Manos \surnamestart Koukoutos\surnameend} \&
  \bibinfo{author}{Viktor \surnamestart Kuncak\surnameend}
  (\bibinfo{year}{2015}): \emph{\bibinfo{title}{Deductive Program Repair}}.
\newblock In \bibinfo{editor}{Daniel \surnamestart Kroening\surnameend} \&
  \bibinfo{editor}{Corina \surnamestart P{\u{a}}s{\u{a}}reanu\surnameend},
  editors: {\sl \bibinfo{booktitle}{Proceedings of the 27th International
  Conference on Computer Aided Verification}}, \bibinfo{publisher}{Springer},
  pp. \bibinfo{pages}{217--233}, \doi{10.1007/978-3-319-21668-3}.
\newblock \urlprefix\url{http://dx.doi.org/10.1007/978-3-319-21668-3_13}.

\bibitemdeclare{inproceedings}{LeonSynth-OOPSLA}
\bibitem{LeonSynth-OOPSLA}
\bibinfo{author}{Etienne \surnamestart Kneuss\surnameend},
  \bibinfo{author}{Ivan \surnamestart Kuraj\surnameend},
  \bibinfo{author}{Viktor \surnamestart Kuncak\surnameend} \&
  \bibinfo{author}{Philippe \surnamestart Suter\surnameend}
  (\bibinfo{year}{2013}): \emph{\bibinfo{title}{Synthesis Modulo Recursive
  Functions}}.
\newblock In: {\sl \bibinfo{booktitle}{Proceedings of the 2013 ACM SIGPLAN
  International Conference on Object Oriented Programming Systems Languages and
  Applications}}, \bibinfo{series}{OOPSLA 2013}, \bibinfo{publisher}{ACM}, pp.
  \bibinfo{pages}{407--426}, \doi{10.1145/2509136.2509555}.
\newblock \urlprefix\url{http://doi.acm.org/10.1145/2509136.2509555}.

\bibitemdeclare{article}{IDAStar}
\bibitem{IDAStar}
\bibinfo{author}{Richard \surnamestart Korf\surnameend} (\bibinfo{year}{1985}):
  \emph{\bibinfo{title}{Depth-first iterative-deepening}}.
\newblock {\sl \bibinfo{journal}{Artificial Intelligence}}
  \bibinfo{volume}{27}(\bibinfo{number}{1}), pp. \bibinfo{pages}{97--109},
  \doi{http://dx.doi.org/10.1016/0004-3702(85)90084-0}.
\newblock
  \urlprefix\url{http://www.sciencedirect.com/science/article/pii/0004370285900840}.

\bibitemdeclare{inproceedings}{LeonSynth-SYNT}
\bibitem{LeonSynth-SYNT}
\bibinfo{author}{Manos \surnamestart Koukoutos\surnameend},
  \bibinfo{author}{Etienne \surnamestart Kneuss\surnameend} \&
  \bibinfo{author}{Viktor \surnamestart Kuncak\surnameend}
  (\bibinfo{year}{2016}): \emph{\bibinfo{title}{An Update on Deductive
  Synthesis and Repair in the Leon Tool}}.
\newblock In: {\sl \bibinfo{booktitle}{Proceedings Fifth Workshop on
  Synthesis}}, \bibinfo{series}{SYNT@CAV 2016}, pp. \bibinfo{pages}{100--111},
  \doi{10.4204/EPTCS.229.9}.
\newblock \urlprefix\url{https://doi.org/10.4204/EPTCS.229.9}.

\bibitemdeclare{article}{Manna-1971}
\bibitem{Manna-1971}
\bibinfo{author}{Zohar \surnamestart Manna\surnameend} \&
  \bibinfo{author}{Richard \surnamestart Waldinger\surnameend}
  (\bibinfo{year}{1971}): \emph{\bibinfo{title}{Toward Automatic Program
  Synthesis}}.
\newblock {\sl \bibinfo{journal}{Communications of the ACM}}
  \bibinfo{volume}{14}(\bibinfo{number}{3}), pp. \bibinfo{pages}{151--165},
  \doi{10.1145/362566.362568}.
\newblock \urlprefix\url{http://doi.acm.org/10.1145/362566.362568}.

\bibitemdeclare{book}{GAMitchell}
\bibitem{GAMitchell}
\bibinfo{author}{Melanie \surnamestart Mitchell\surnameend}
  (\bibinfo{year}{1998}): \emph{\bibinfo{title}{An Introduction to Genetic
  Algorithms}}.
\newblock \bibinfo{publisher}{MIT Press}.

\bibitemdeclare{inproceedings}{Osera-Myth-15}
\bibitem{Osera-Myth-15}
\bibinfo{author}{Peter-Michael \surnamestart Osera\surnameend} \&
  \bibinfo{author}{Steve \surnamestart Zdancewic\surnameend}
  (\bibinfo{year}{2015}): \emph{\bibinfo{title}{Type-and-example-directed
  Program Synthesis}}.
\newblock In: {\sl \bibinfo{booktitle}{Proceedings of the 36th ACM SIGPLAN
  Conference on Programming Language Design and Implementation}},
  \bibinfo{publisher}{ACM}, pp. \bibinfo{pages}{619--630},
  \doi{10.1145/2737924.2738007}.
\newblock \urlprefix\url{http://doi.acm.org/10.1145/2737924.2738007}.

\bibitemdeclare{inproceedings}{PeiETAL15RepairTests}
\bibitem{PeiETAL15RepairTests}
\bibinfo{author}{Yu~\surnamestart Pei\surnameend}, \bibinfo{author}{Carlo~A.
  \surnamestart Furia\surnameend}, \bibinfo{author}{Mart{\'{\i}}n \surnamestart
  Nordio\surnameend} \& \bibinfo{author}{Bertrand \surnamestart
  Meyer\surnameend} (\bibinfo{year}{2015}): \emph{\bibinfo{title}{Automated
  Program Repair in an Integrated Development Environment}}.
\newblock In: {\sl \bibinfo{booktitle}{37th {IEEE/ACM} International Conference
  on Software Engineering, {ICSE} 2015, Florence, Italy, May 16-24, 2015,
  Volume 2}}, pp. \bibinfo{pages}{681--684}, \doi{10.1109/ICSE.2015.222}.
\newblock \urlprefix\url{https://doi.org/10.1109/ICSE.2015.222}.

\bibitemdeclare{inproceedings}{Chlorophyll}
\bibitem{Chlorophyll}
\bibinfo{author}{Phitchaya~Mangpo \surnamestart Phothilimthana\surnameend},
  \bibinfo{author}{Tikhon \surnamestart Jelvis\surnameend},
  \bibinfo{author}{Rohin \surnamestart Shah\surnameend},
  \bibinfo{author}{Nishant \surnamestart Totla\surnameend},
  \bibinfo{author}{Sarah \surnamestart Chasins\surnameend} \&
  \bibinfo{author}{Rastislav \surnamestart Bodik\surnameend}
  (\bibinfo{year}{2014}): \emph{\bibinfo{title}{Chlorophyll: {S}ynthesis-aided
  Compiler for Low-power Spatial Architectures}}.
\newblock In: {\sl \bibinfo{booktitle}{Proceedings of the 35th ACM SIGPLAN
  Conference on Programming Language Design and Implementation}},
  \bibinfo{series}{PLDI '14}, \bibinfo{publisher}{ACM}, pp.
  \bibinfo{pages}{396--407}, \doi{10.1145/2594291.2594339}.
\newblock \urlprefix\url{http://doi.acm.org/10.1145/2594291.2594339}.

\bibitemdeclare{inproceedings}{Synquid}
\bibitem{Synquid}
\bibinfo{author}{Nadia \surnamestart Polikarpova\surnameend},
  \bibinfo{author}{Ivan \surnamestart Kuraj\surnameend} \&
  \bibinfo{author}{Armando \surnamestart Solar-Lezama\surnameend}
  (\bibinfo{year}{2016}): \emph{\bibinfo{title}{Program Synthesis from
  Polymorphic Refinement Types}}.
\newblock In: {\sl \bibinfo{booktitle}{Proceedings of the 37th ACM SIGPLAN
  Conference on Programming Language Design and Implementation}},
  \bibinfo{series}{PLDI '16}, \bibinfo{publisher}{ACM}, pp.
  \bibinfo{pages}{522--538}, \doi{10.1145/2908080.2908093}.
\newblock \urlprefix\url{http://doi.acm.org/10.1145/2908080.2908093}.

\bibitemdeclare{inproceedings}{Vechev-PHOG0}
\bibitem{Vechev-PHOG0}
\bibinfo{author}{Veselin \surnamestart Raychev\surnameend},
  \bibinfo{author}{Pavol \surnamestart Bielik\surnameend} \&
  \bibinfo{author}{Martin \surnamestart Vechev\surnameend}
  (\bibinfo{year}{2016}): \emph{\bibinfo{title}{Probabilistic Model for Code
  with Decision Trees}}.
\newblock In: {\sl \bibinfo{booktitle}{Proceedings of the 2016 ACM SIGPLAN
  International Conference on Object-Oriented Programming, Systems, Languages,
  and Applications}}, \bibinfo{series}{OOPSLA 2016}, \bibinfo{publisher}{ACM},
  pp. \bibinfo{pages}{731--747}, \doi{10.1145/2983990.2984041}.
\newblock \urlprefix\url{http://doi.acm.org/10.1145/2983990.2984041}.

\bibitemdeclare{inproceedings}{Vechev-BigCode}
\bibitem{Vechev-BigCode}
\bibinfo{author}{Veselin \surnamestart Raychev\surnameend},
  \bibinfo{author}{Martin \surnamestart Vechev\surnameend} \&
  \bibinfo{author}{Andreas \surnamestart Krause\surnameend}
  (\bibinfo{year}{2015}): \emph{\bibinfo{title}{Predicting Program Properties
  from "Big Code"}}.
\newblock In: {\sl \bibinfo{booktitle}{Proceedings of the 42nd Annual ACM
  SIGPLAN-SIGACT Symposium on Principles of Programming Languages}},
  \bibinfo{series}{POPL 2015}, \bibinfo{publisher}{ACM}, pp.
  \bibinfo{pages}{111--124}, \doi{10.1145/2676726.2677009}.
\newblock \urlprefix\url{http://doi.acm.org/10.1145/2676726.2677009}.

\bibitemdeclare{inproceedings}{Vechev-PLDI14}
\bibitem{Vechev-PLDI14}
\bibinfo{author}{Veselin \surnamestart Raychev\surnameend},
  \bibinfo{author}{Martin \surnamestart Vechev\surnameend} \&
  \bibinfo{author}{Eran \surnamestart Yahav\surnameend} (\bibinfo{year}{2014}):
  \emph{\bibinfo{title}{Code Completion with Statistical Language Models}}.
\newblock In: {\sl \bibinfo{booktitle}{Proceedings of the 35th ACM SIGPLAN
  Conference on Programming Language Design and Implementation}},
  \bibinfo{series}{PLDI '14}, \bibinfo{publisher}{ACM}, pp.
  \bibinfo{pages}{419--428}, \doi{10.1145/2594291.2594321}.
\newblock \urlprefix\url{http://doi.acm.org/10.1145/2594291.2594321}.

\bibitemdeclare{inproceedings}{CVC4-SyGuS}
\bibitem{CVC4-SyGuS}
\bibinfo{author}{Andrew \surnamestart Reynolds\surnameend},
  \bibinfo{author}{Morgan \surnamestart Deters\surnameend},
  \bibinfo{author}{Viktor \surnamestart Kuncak\surnameend},
  \bibinfo{author}{Cesare \surnamestart Tinelli\surnameend} \&
  \bibinfo{author}{Clark \surnamestart Barrett\surnameend}
  (\bibinfo{year}{2015}): \emph{\bibinfo{title}{Counterexample-Guided
  Quantifier Instantiation for Synthesis in {SMT}}}.
\newblock In \bibinfo{editor}{Daniel \surnamestart Kroening\surnameend} \&
  \bibinfo{editor}{Corina \surnamestart P{\u{a}}s{\u{a}}reanu\surnameend},
  editors: {\sl \bibinfo{booktitle}{27th International Conference on Computer
  Aided Verification}}, \bibinfo{publisher}{Springer}, pp.
  \bibinfo{pages}{198--216}, \doi{10.1007/978-3-319-21668-3}.
\newblock \urlprefix\url{http://dx.doi.org/10.1007/978-3-319-21668-3_12}.

\bibitemdeclare{inproceedings}{SMAStar}
\bibitem{SMAStar}
\bibinfo{author}{Stuart \surnamestart Russell\surnameend}
  (\bibinfo{year}{1992}): \emph{\bibinfo{title}{Efficient Memory-bounded Search
  Methods}}.
\newblock In: {\sl \bibinfo{booktitle}{Proceedings of the 10th European
  Conference on Artificial Intelligence}}, \bibinfo{series}{ECAI '92},
  \bibinfo{publisher}{John Wiley \& Sons, Inc.}, pp. \bibinfo{pages}{1--5}.
\newblock \urlprefix\url{http://dl.acm.org/citation.cfm?id=145448.145476}.

\bibitemdeclare{book}{AIMA}
\bibitem{AIMA}
\bibinfo{author}{Stuart \surnamestart Russell\surnameend} \&
  \bibinfo{author}{Peter \surnamestart Norvig\surnameend}
  (\bibinfo{year}{2009}): \emph{\bibinfo{title}{Artificial Intelligence: {A}
  Modern Approach}}, \bibinfo{edition}{3rd} edition.
\newblock \bibinfo{publisher}{Pearson}.

\bibitemdeclare{inproceedings}{Sketch}
\bibitem{Sketch}
\bibinfo{author}{Armando \surnamestart Solar-Lezama\surnameend},
  \bibinfo{author}{Liviu \surnamestart Tancau\surnameend},
  \bibinfo{author}{Rastislav \surnamestart Bodik\surnameend},
  \bibinfo{author}{Sanjit \surnamestart Seshia\surnameend} \&
  \bibinfo{author}{Vijay \surnamestart Saraswat\surnameend}
  (\bibinfo{year}{2006}): \emph{\bibinfo{title}{Combinatorial Sketching for
  Finite Programs}}.
\newblock In: {\sl \bibinfo{booktitle}{Proceedings of the 12th International
  Conference on Architectural Support for Programming Languages and Operating
  Systems}}, \bibinfo{series}{ASPLOS XII}, \bibinfo{publisher}{ACM}, pp.
  \bibinfo{pages}{404--415}, \doi{10.1145/1168857.1168907}.
\newblock \urlprefix\url{http://doi.acm.org/10.1145/1168857.1168907}.

\bibitemdeclare{inproceedings}{Transit}
\bibitem{Transit}
\bibinfo{author}{Abhishek \surnamestart Udupa\surnameend},
  \bibinfo{author}{Arun \surnamestart Raghavan\surnameend},
  \bibinfo{author}{Jyotirmoy \surnamestart Deshmukh\surnameend},
  \bibinfo{author}{Sela \surnamestart Mador-Haim\surnameend},
  \bibinfo{author}{Milo \surnamestart Martin\surnameend} \&
  \bibinfo{author}{Rajeev \surnamestart Alur\surnameend}
  (\bibinfo{year}{2013}): \emph{\bibinfo{title}{TRANSIT: {S}pecifying Protocols
  with Concolic Snippets}}.
\newblock In: {\sl \bibinfo{booktitle}{Proceedings of the 34th ACM SIGPLAN
  Conference on Programming Language Design and Implementation}},
  \bibinfo{series}{PLDI '13}, \bibinfo{publisher}{ACM}, pp.
  \bibinfo{pages}{287--296}, \doi{10.1145/2491956.2462174}.
\newblock \urlprefix\url{http://doi.acm.org/10.1145/2491956.2462174}.

\end{thebibliography}

\end{document}